\documentclass[draftcls,onecolumn]{IEEEtran}%one column

\usepackage{amsmath,amsfonts}
\usepackage{algorithmicx}
\usepackage{algorithm}
\usepackage{array}
\usepackage{subfigure}
\usepackage{textcomp}
\usepackage{stfloats}
\usepackage{url}
\usepackage{verbatim}
\usepackage{graphicx}
\usepackage{cite}
\usepackage{pdfsync}
\usepackage{amsthm}
\usepackage{times,amssymb,float,nicefrac,mathrsfs,float}
\usepackage{bm}
\usepackage{bbm}
\usepackage{algpseudocode}
\usepackage{color}
\usepackage{booktabs}
\usepackage{epstopdf}

\newtheorem{theorem}{Theorem}
\newtheorem{corollary}{Corollary}
\newtheorem{definition}{Definition}
\newtheorem{proposition}{Proposition}
\newtheorem{example}{Example}
\newtheorem{remark}{Remark}
\newtheorem{lemma}{Lemma}

\begin{document}

\title{Derivative Descendants of Cyclic Codes and Derivative Decoding
}

\author{Qin Huang$^*$,~\IEEEmembership{Senior Member,~IEEE,}
\and 
Bin Zhang

}
\maketitle

{\renewcommand{\thefootnote}{}\footnotetext{

\vspace{-.2in}
 
\noindent\rule{1.5in}{.4pt}

Part of this article was presented at GlobeCom 2022.
This work was supported by the National Natural Science Foundation
of China under Grant 62071026.
(Qin Huang and Bin Zhang contributed equally to this work.)
(Corresponding author: Qin Huang.)
}

\begin{abstract}
This paper defines cyclic and minimal \emph{derivative descendants} (DDs) of an extended cyclic code from the derivative of the Mattson-Solomon polynomials, respectively. First, it demonstrates that the cyclic DDs are the same extended cyclic code. It allows us to perform soft-decision decoding for extended cyclic codes based on their cyclic DDs. Then, it proves that the minimal DDs are equivalent codes. It also allows us to perform soft-decision decoding based on the minimal DDs with permutations. Simulation results show that our proposed derivative decoding can be close to the maximum likelihood decoding for certain extended cyclic codes, including some extended BCH codes. 
\end{abstract}


\section{Introduction}
Cyclic codes, first studied in 1957 \cite{prange1957cyclic}, form a large class of error-control codes which include many well-known codes, 
e.g.,  \emph{Bose-Chaudhuri-Hocquenghem} (BCH) codes, \emph{Reed-Solomon} codes, finite geometry codes, punctured \emph{Reed-Muller} (RM) codes etc. \cite{macwilliams1977theory,lin2001error,blahut2003algebraic,berlekamp2015algebraic}. 
Due to the cyclic structure, their encoding and hard-decision decoding can be implemented efficiently. Moreover, their inherent algebraic structure and soft-decision decoding \cite{vardy1994maximum,kamiya2001algebraic,bossert2022hard,tapp1999extended,lin2020scheme,kou2001low,sidel1992decoding,dumer2006soft,ye2020recursive,lian2020decoding} have always attracted a lot of attention.

This paper starts from the derivative of Mattson-Solomon (MS) polynomials \cite{macwilliams1977theory,mattson1961new}. We define two types of derivative descendants (DDs) of an extended cyclic code $\mathcal{C}$, cyclic DDs and minimal DDs, respectively. The first type is defined as the smallest extended cyclic codes containing the derivatives of all the codewords in $\mathcal{C}$. The second type is defined as the smallest subspaces consisting of the derivatives of all codewords in $\mathcal{C}$.

First, we demonstrate that the cyclic DDs of an extended binary cyclic code in different directions result in the same extended cyclic code. It can be specified by analyzing the exponent set of MS polynomials. Based on their cyclic DDs, we propose a soft-decision derivative decoding algorithm for extended binary cyclic codes. It consists of three steps: calculating log-likelihood ratios (LLRs) of cyclic DDs, decoding cyclic DDs and voting for decision. Simulation results show that the performance of the proposed derivative decoding is close to that of the maximum likelihood decoding (MLD) for cyclic codes, e.g., $(64,45)$ and $(64,24)$ extended BCH (eBCH) codes. Besides we conversely introduce cyclic derivative ascendant (DA) of an extended cyclic code $\mathcal{C}$ as well as their decoding. 

Then, we prove that the minimal DDs of an extended binary cyclic code in different directions are equivalent. Moreover, it reveals that the cyclic shift of a codeword in a minimal DD is a codeword in another minimal DD. As a result, the derivative decoding can be carried out with only one decoder for the minimal DD in one direction and cyclic shifting. Due to the small dimension of minimal DDs, it is attractive to perform derivative decoding based on ordered statistics decoding (OSD) \cite{fossorier1995soft}. Simulation results show that the derivative decoding based on the OSD with order-1 can outperform the higher order OSD.

The rest of the paper is organized as follows. 
Section \ref{sec:preli} gives a brief review of cyclic codes and MS polynomials.
In Section \ref{sec:code construction}, we define the cyclic DDs and cyclic DAs of extended cyclic codes. Section \ref{sec:decoding algorithm} presents the derivative decoding algorithm. In Section \ref{sec:MDD}, we define the minimal DDs and present the derivative decoding based on decodings of minimal DDs. Section \ref{sec:conclusion} concludes this paper.

\section{Cyclic Codes and Their Decomposition}\label{sec:preli}

\subsection{Cyclic codes and Mattson-Solomon polynomials}

A linear code $\mathcal{C}$ of length $n$ is \emph{cyclic} if a cyclic shift of any codeword is also a codeword, i.e. whenever $\bm{a}=[a_i, i\in [n]]$ is in $\mathcal{C}$ then so is $[a_{i+1}, i\in [n]]$.
Here, $[n]\triangleq \{0, 1, 2, ..., n-1\}$ and subscripts are reduced modulo $n$.

Let $m$ be a positive integer. A binary cyclic code $\mathcal{C}$ of length $n=2^m-1$ and dimension $0<k\leq n$ is an ideal in the ring $\mathbb{F}_2[x]/(x^n-1)$, which is generated by a generator polynomial $g(x)$ with degree $n-k$ such that $g(x)$ divides $x^n-1$. 

Let $\alpha$ denote a primitive element of $\mathbb{F}_{2^m}$.
For a codeword $\bm{a}=[a_0, a_1, ..., a_{n-1}]$ corresponding to a code polynomial $a(x)=\sum_{i=0}^{n-1}a_{i}x^i$, the associated Mattson-Solomon polynomial is defined over $\mathbb{F}_{2^m}$ as follows
\begin{equation*}
A(z) \triangleq \sum_{j=0}^{n-1}A_{j}z^{j},
\end{equation*}
where 
\begin{equation*}
A_{j} = a(\alpha^{-j}) = \sum_{i=0}^{n-1}a_i \alpha^{-ij}.
\end{equation*}
\IEEEpubidadjcol
The codeword $\bm{a}$ can be recovered from $A(z)$ by 
\begin{equation*}
\bm{a} = [a_{i}, i\in [n]] = [A(\alpha^i), i \in [n]].
\end{equation*}
The coefficient $A_{j}$ is fixed to 0 if and only if $\alpha^{-j}$ is a zero of $g(x)$. 
Moreover, the MS polynomial of the cyclic shift of $\bm{a}$ is $A(\alpha z)$.

We define the \emph{exponent set} of all the MS polynomials associated with $\mathcal{C}$ as follows
\begin{equation}\label{eq:exponent set}
S_\mathcal{C} \triangleq \{j\in [n]: g(\alpha^{-j})\neq 0\}.
\end{equation}
For brevity, we call it the exponent set of $\mathcal{C}$. Please note that its size is the same as the dimension $k$. We can express $\mathcal{C}$ as
\begin{equation*}
\mathcal{C} = \{[A(\alpha^i), i\in [n]]: A(z) = \sum_{j \in S_{\mathcal{C}}}A_{j}z^j\},
\end{equation*}
where $A_j \in \mathbb{F}_{2^m}$. The \emph{conjugacy constraint} \cite[Ch. 6]{blahut2003algebraic}, i.e. $A_{2j}=A_j^2$ is required to keep $[A(\alpha^i), i\in [n]]$ binary.

The cyclic code $\mathcal{C}$ can be extended by adding an overall parity-check bit to each codeword.
The overall parity-check bit of a codeword $\bm{a}$ is the evaluation of the corresponding MS polynomial at $0$, i.e., $A(0)$ \cite[Ch. 8]{macwilliams1977theory}.
Therefore, the extended cyclic code of $\mathcal{C}$ can be also identified by $S_{\mathcal{C}}$.
We denote $0$ in $\mathbb{F}_{2^m}$ by $\alpha^{\infty}$ and define $I \triangleq \{\infty\} \cup [n]$.
The extended cyclic code $\mathcal{C}$ with exponent set $S_{\mathcal{C}}$ can be expressed as
\begin{equation*}
\mathcal{C} = \{[A(\alpha^i), i\in I]: A(z) = \sum_{j \in S_{\mathcal{C}}}A_{j}z^j\}.
\end{equation*}
Please note that $\alpha^\infty \alpha = \alpha^\infty$.
Thus we make the agreement $\infty + 1=\infty\text{ mod }n$.
Then $\mathcal{C}$ is an extended cyclic code if whenever $\bm{a}=[a_i, i\in I]$ is in $\mathcal{C}$ then its cyclic shift $[a_{i+1}, i\in I]$ is also in $\mathcal{C}$.  
In the following, we mainly focus on extended cyclic codes, and may use $\bm{a}$ and $A(z)$ to denote the codeword interchangeably.

\subsection{Decomposing cyclic codes as a direct sum of minimal cyclic codes}

For an integer $s \in [n]$, the cyclotomic coset modulo $n$ containing $s$ is $C_s \triangleq \{s, 2s, 2^2s, ..., 2^{m_s-1}s\}$, where $m_s$ is the smallest positive integer such that $2^{m_s}s = s$ \text{mod} $n$.
The smallest entry of $C_s$ is called the \emph{coset representative}. 
For a subset $S$ of $[n]$, we denote the smallest and the largest element of $S$ by $\texttt{min}(S)$ and $\texttt{max}(S)$, respectively.
We denote the union of all the cyclotomic cosets which have nonempty intersections with $S$ as $\texttt{cc}(S) \triangleq \bigcup_{s\in S}C_s$. 
And we denote the set consisting of all the coset representatives in $\texttt{cc}(S)$ as $\texttt{cr}(S) \triangleq \bigcup_{s\in S}\{\texttt{min}(C_s)\}$.

The extended \emph{minimal cyclic code} associated with the cyclotomic coset $C_s$ is
\begin{equation*}
\begin{aligned}
\mathcal{M}_s = \{[A(\alpha^i), i\in I] : A(z)=T_{m_s}(A_sz^s) \\ \text{ for all } A_{s}\in \mathbb{F}_{2^{m_s}} \},
\end{aligned}
\end{equation*} 
where $T_{m_s}(z)$ is the \emph{trace function}
\begin{equation*}
T_{m_s}(z) \triangleq \sum_{j\in[m_s]}z^{2^j},
\end{equation*}
and $\mathbb{F}_{2^{m_s}}$ is a subfield of $\mathbb{F}_{2^m}$.
It is clear that the exponent set of $\mathcal{M}_s$ is $C_s$.

An extended cyclic code with exponent set $S_{\mathcal{C}}$ can be expressed as a direct sum of the extended minimal cyclic codes, i.e.,
\begin{equation*}
\begin{aligned}
\mathcal{C} & = \bigoplus_{s\in \texttt{cr}(S_{\mathcal{C}})} \mathcal{M}_s \\
  & = \{[A(\alpha^i), i\in I] : A(z)=\sum_{s\in \texttt{cr}(S_{\mathcal{C}})}T_{m_s}(A_sz^s) \\
 & \text{ for all } A_s\in \mathbb{F}_{2^{m_s}} \},
\end{aligned}
\end{equation*}
where $\bigoplus$ is the direct sum operator. 
We call the set $\texttt{cr}(S_{\mathcal{C}})$ as the \emph{representative set} of $S_{\mathcal{C}}$. We end this section with the following example.
\begin{example}\label{ex:ex_1}
Let $\alpha$ denote a primitve element in $\mathbb{F}_{2^4}$.
Consider the $(16,7)$ extended cyclic code $\mathcal{C}$ associated with the generator polynomial $g(x)=1+x^4+x^6+x^7+x^8$.
The zeros of $g(x)$ are $\alpha^1, \alpha^2, \alpha^4, \alpha^8, \alpha^3, \alpha^6, \alpha^{12}, \alpha^9$.
From (\ref{eq:exponent set}), the exponent set of $\mathcal{C}$ is $S_{\mathcal{C}}=\{0, 1, 2, 4, 8, 5, 10\}$.
Then $\mathcal{C}$ can be identified by the set 
$$\{[A(\alpha^i), i\in I]: A(z) = \sum_{j \in S^{\mathcal{C}}}A_{j}z^j\},$$
where $A_j\in\mathbb{F}_{2^4}$ and satisfies $A_{2j}=A_j^2$.
The representative set of $S_\mathcal{C}$ is $\{0, 1, 5\}$.
Note that $m_0=1$, $m_1=4$, $m_5=2$. 
Then $\mathcal{C}$ can be expressed as
\begin{equation*}
\begin{aligned}
\mathcal{C}=\{[A(\alpha^i), i\in I]:A_0 + T_4(A_1z) + T_2(A_5z^5)\},
\end{aligned}
\end{equation*}
where $A_0\in \mathbb{F}_2$, $A_1 \in \mathbb{F}_{2^4}$, $A_5 \in \mathbb{F}_{2^2}$.
\end{example}

\section{Cyclic Derivative Descendants and Ascendants}\label{sec:code construction}

This section introduces cyclic DDs and cyclic DAs of extended cyclic codes. Their dimensions and distances are also investigated.

\subsection{Cyclic derivative descendants}
Let $\beta$ be a power of $\alpha$. 
Consider a codeword $\bm{a}$ and its MS polynomial $A(z)$.
The derivative of $A(z)$ in the direction $\beta$ is defined as 
\begin{equation}\label{eq:derivative}
\Delta_{\beta} A(z) \triangleq A(z+\beta) - A(z).
\end{equation}
With the above definition, we define the cyclic DDs of extended cyclic codes.

\begin{definition}\label{def:cyclic derivative descendant}
For an extended cyclic code $\mathcal{C}$, its \emph{cyclic derivative descendant} in the direction $\beta$ denoted by  $\mathcal{D}(\mathcal{C}, \beta)$ is the extended cyclic code with the smallest dimension which contains
\begin{equation*}
\{[\Delta_{\beta}A(\alpha^i), i \in I]: A(z)\in \mathcal{C} \}.
\end{equation*}
\end{definition}
In the following, we may use $\Delta_{\beta}A(z)$ to denote the vector $[\Delta_{\beta}A(\alpha^i), i \in I]$ if the context is clear.
For simplicity, we may call the vector the derivative of $\bm{a}$.

In fact, the cyclic DDs in all the directions are the same.
To prove this, we start with the extended minimal cyclic codes. 
For an integer $s \in [n]$, we denote its binary expansion by $\overline{s}=[s_0, s_1, ..., s_{m-1}]$ such that $s = \sum_{j = 0}^{m-1}s_j2^{j}$.
We define the support set of the binary expansion of $s$ as $W_s \triangleq \{j \in [m]: s_j \neq 0\}$.
And we say the binary expansion of $s'$ is properly covered by that of $s$ if $W_{s'}\subsetneqq W_s$.
Let $P(s)$ denote the set consisting of all the nonnegetive intergers whose binary expansion is properly covered by that of $s$, i.e.
\begin{equation}\label{eq:Ds}
P(s) \triangleq \{\sum_{j\in V}2^j: V\subsetneqq W_s\}.
\end{equation}
The exponent set of the cyclic DDs of an extended minimal cyclic code is given by the following lemma.
\begin{figure*}[htbp]
\centering
\includegraphics[width=0.8\textwidth]{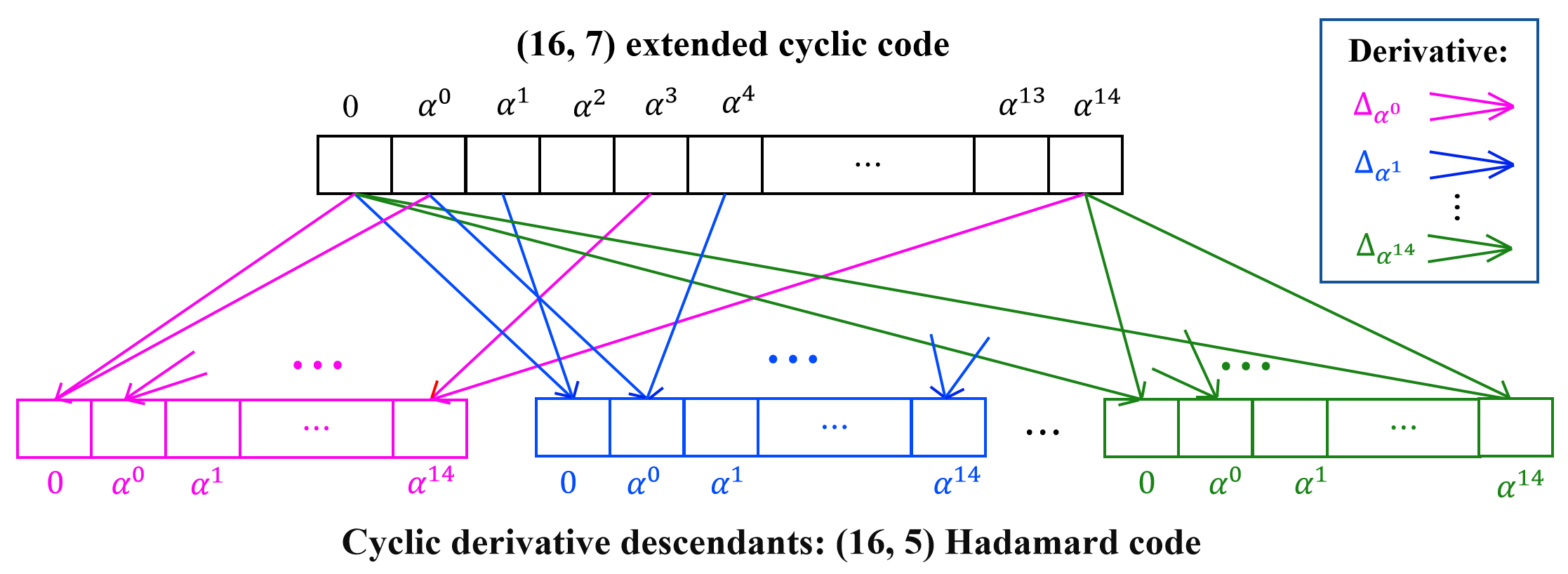}
\caption{The cyclic DD of the $(16, 7)$ extended cyclic code is a $(16, 5)$ extended cyclic code.}
\label{fig:descendantExample}
\end{figure*}

\begin{lemma}\label{lemma: 1}
Consider the extended minimal cyclic code $\mathcal{M}_s$. 
The exponent set of $\mathcal{D}(\mathcal{M}_s, \beta)$ for any $\beta$ is $\texttt{cc}\Big(P(s)\Big)$.
\end{lemma}

\begin{proof}
For any $A(z) \in \mathcal{M}_s$, its derivative in direction $\beta$ is
\begin{equation*}
\begin{aligned}
\Delta_{\beta}A(z) & = T_{m_s}\Big(A_s(z+\beta)^s\Big) - T_{m_s}(A_sz^s) \\
& = T_{m_s}\Bigg(A_s\Big((z+\beta)^s-z^s\Big)\Bigg).
\end{aligned}
\end{equation*}
Note that
\begin{equation*}
\begin{aligned}
(z+\beta)^s -z^s  & = \prod_{j\in W_s}(z^{2^j}+\beta^{2^j}) - \prod_{j\in W_s}z^{2^j} \\
& = \sum_{V\subsetneqq W_s} ( z^{\sum_{j\in V}2^j} \beta^{\sum_{j\in W_s/V}2^j} ) \\
& = \sum_{k\in P(s)} z^k \beta^{s-k}.
\end{aligned}
\end{equation*}
Then 
\begin{equation*}
\begin{aligned}
\Delta_{\beta}A(z) & = T_{m_s}(A_s\sum_{k\in P(s)} z^k \beta^{s-k}).
\end{aligned}
\end{equation*}
From the above equation, we see that the exponents of $z$ must be a subset of $\texttt{cc}\Big(P(s)\Big)$.
Note that the coefficients of $\Delta_{\beta}A(z)$ must satisfy the conjugacy constraint, because $\Delta_{\beta}A(\alpha^i) = A(\alpha^i+\beta)-A(\alpha^i)$, and $A(\alpha^i)$ and $A(\alpha^i+\beta)$ are in $\mathbb{F}_2$ for all $i \in I$.
Therefore we can write  $\Delta_{\beta}A(z)$ in the form 
\begin{equation*}
\Delta_{\beta}A(z) = \sum_{s' \in \texttt{cr}\Big(P(s)\Big)}T_{m_{s'}}(A_{s'}'z^{s'}),
\end{equation*}
where
\begin{equation*}
A_{s'}' = 
\sum_{i\in[m_s],k\in P(s), \atop k2^i[\text{mod $n$}]=s'} \beta^{(s-k)2^i} A_s^{2^i}.
\end{equation*}
Treat $A_{s'}'$ as a function of $A_s$, i.e. $A_{s'}'(A_s)$.
Note that the degree of $A_{s'}'(A_s)$ is at most $2^{m_s-1}$ which implies $A_{s'}'(A_s)$ has at most $2^{m_s-1}$ roots.
Therefore, $A_{s'}'$ is not always zero.
As a result, the representative set of the exponent set of $\mathcal{D}(\mathcal{C}, \beta)$ is exactly $\texttt{cr}\Big(P(s)\Big)$ while the exponent set is $\texttt{cc}\Big(P(s)\Big)$.
\end{proof}

The above lemma shows that the cyclic DDs of an extended minimal cyclic code in different directions are the same.
Using the fact that the extended cyclic code $\mathcal{C}$ is a direct sum of extended minimal cyclic codes, we obtain the following theorem immediately.
\begin{theorem}\label{thm:DSN of DD}
For an extended cyclic code $\mathcal{C}$ with exponent set $S_{\mathcal{C}}$,
its cyclic DDs in different directions are the same code denoted by $\mathcal{D}(\mathcal{C})$,
whose exponent set $S_{\mathcal{D}}$ is $\bigcup_{s\in \texttt{cr}(S_{\mathcal{C}})}\texttt{cc}\Big(P(s)\Big)$ and the corresponding representative set is $\bigcup_{s\in \texttt{cr}(S_{\mathcal{C}})}\texttt{cr}\Big(P(s)\Big)$.
\end{theorem}

\begin{example}\label{ex:ex_2}
Continuation of Example \ref{ex:ex_1}. 
The representative set of the exponent set of the $(16, 7)$ extended cyclic code $\mathcal{C}$ is $\{0, 1, 5\}$.
According to (\ref{eq:Ds}), $P(0) = \emptyset$, $P(1) = \{0\}$, $P(5) = \{0, 1, 4\}$.
From Theorem \ref{thm:DSN of DD}, we conclude that 
the cyclic DD of $\mathcal{C}$ denoted by $\mathcal{D}(\mathcal{C})$ is the $(16, 5)$ extended cyclic code with the exponent set $S_{\mathcal{D}} = \{0, 1, 2, 4, 8\}$, i.e., 
$$\mathcal{D}(\mathcal{C})=\{[A(\alpha^i), i \in I]: A(z) = A_0 + T_4(A_1z)\}.$$ 
It shows that $\mathcal{D}(\mathcal{C})$ is the $(16, 5)$ Hadamard code.
For code codeword $\bm{a} = [A(\alpha^i), i\in I] \in \mathcal{C}$, all its derivatives, 
\begin{equation*}
\begin{aligned}
&  [\Delta_{\alpha^0}A(\alpha^i), i\in I] \\
&  [\Delta_{\alpha^1}A(\alpha^i), i\in I] \\ 
& ... \\
&  [\Delta_{\alpha^{14}}A(\alpha^i), i\in I],
\end{aligned} 
\end{equation*} 
are codewords in $\mathcal{D}(\mathcal{C})$ as illustrated in Fig. \ref{fig:descendantExample}.
\end{example}

For any nontrivial extended binary cyclic code $\mathcal{C}$, i.e. $S_{\mathcal{C}}\neq \{0\}$, we give the following propositions to characaterize the dimension and distance of their cyclic DDs. 

For any binary vector $\bm{v}$, we denote its Hamming weight by $\texttt{wt}(\bm{v})$.
For a subset $S$ of $[n]$, we define $\texttt{deg}(S) \triangleq \texttt{max}(\bigcup_{s\in S}\{\texttt{wt}(\overline{s})\})$.
Let $d$ denote the minimum Hamming distance of $\mathcal{C}$.
And let $k_{\mathcal{D}}$ and $d_{\mathcal{D}}$ denote the dimension and the minimum Hamming distance of $\mathcal{D}(\mathcal{C})$, respectively.

\begin{proposition}\label{prop:dd dimension}
$k_{\mathcal{D}} \leq \sum_{i=0}^{\texttt{\emph{deg}}(S_{\mathcal{C}})-1}\binom{m}{i}$.
\end{proposition}
\begin{proof}

The dimension of $\mathcal{D}(\mathcal{C})$ satisfies
$$k_{\mathcal{D}} = |S_{\mathcal{D}}| \leq \sum_{i=0}^{\texttt{deg}(S_{\mathcal{D}})}\binom{m}{i}.$$
From Theorem \ref{thm:DSN of DD},
\begin{equation*}
\begin{aligned}
\texttt{deg}(S_{\mathcal{D}}) = \texttt{deg}\Bigg(\bigcup_{s\in \texttt{cr}(S_{\mathcal{C}})}\texttt{cc}\Big(P(s)\Big)\Bigg).
\end{aligned}
\end{equation*}
Note that the binary expansion of $2s$ modulo $n$ is a cyclic shift of $\overline{s}$.
Thus, $\texttt{deg}(C_s) = \texttt{wt}(\overline{s})$ and $\texttt{deg}\Bigg(\texttt{cc}\Big(P(s)\Big)\Bigg) = \texttt{deg}\Big(P(s)\Big)$.
Then, 
\begin{equation*}
\begin{aligned}
\texttt{deg}(S_{\mathcal{D}}) = \texttt{deg}\Big(\bigcup_{s\in \texttt{cr}(S_{\mathcal{C}})}P(s)\Big).
\end{aligned}
\end{equation*}
Note that $\bigcup_{s\in \texttt{cr}(S_{\mathcal{C}})}P(s) \subseteq \bigcup_{s\in S_{\mathcal{C}}}P(s)$, then 
\begin{equation*}
\begin{aligned}
\texttt{deg}(S_{\mathcal{D}}) \leq \texttt{deg}\Big(\bigcup_{s\in S_{\mathcal{C}}}P(s)\Big).
\end{aligned}
\end{equation*}
For $s=0$, we have $W_s=\emptyset$ and $P(s)=\emptyset$. Then $\texttt{deg}\Big(P(s)\Big)=0$.
For any positive $s\in S_{\mathcal{C}}$, from (\ref{eq:Ds}), we have $\texttt{deg}\Big(P(s)\Big) = \texttt{wt}(\overline{s})-1$.
Then $\texttt{deg}\Big(\bigcup_{s\in S_{\mathcal{C}}}P(s)\Big) =  \texttt{deg}(S_{\mathcal{C}})-1$.
As a result,
$$k_{\mathcal{D}} \leq \sum_{i=0}^{\texttt{deg}(S_{\mathcal{D}})}\binom{m}{i} \leq \sum_{i=0}^{\texttt{deg}(S_{\mathcal{C}})-1}\binom{m}{i}.$$
\end{proof}

\begin{proposition}\label{prop:distance of dd}
$d_{\mathcal{D}} \leq 2d$.
\end{proposition}

\begin{proof}
Consider the codeword $A(z)\in \mathcal{C}$ with $\texttt{wt}\Big(A(z)\Big) = d$.
The Hamming weight of its derivative $\Delta_{\beta}A(z)$ satisifies
\begin{equation*}
\begin{aligned}
\texttt{wt}\Big(\Delta_{\beta}A(z)\Big) = & \texttt{wt}\Big(A(z+\beta)-A(z)\Big) \\
\leq & \texttt{wt}\Big(A(z+\beta)\Big)+ \texttt{wt}\Big(A(z)\Big) \\
= & 2d.
\end{aligned}
\end{equation*}
As a result, $d_{\mathcal{D}} \leq 2d$.
\end{proof}

\subsection{Cyclic derivative ascendant}

Conversely to the cyclic DDs, we define cyclic DAs of an extended cyclic code as follows.

\begin{definition}\label{def:derivative ascendant}
For an extended cyclic code $\mathcal{C}$, we define its cyclic \emph{derivative ascendant} denoted by $\mathcal{A}(\mathcal{C})$ as the extended cyclic code with the largest dimension
such that $\mathcal{D}\Big(\mathcal{A}(\mathcal{C})\Big) \subseteq \mathcal{C}$.
\end{definition}

We give a proposition to characterize the exponent set of the cyclic DA of $\mathcal{C}$.
\begin{proposition}\label{prop:ascendant DS}
Let $S_{\mathcal{A}}$ denote the exponent set of $\mathcal{A}(\mathcal{C})$. A nonnegative integer $s$ smaller than $n$ is in $S_{\mathcal{A}}$ if and only if
\begin{equation*}
\texttt{cc}\Big(P(s)\Big) \subseteq S_{\mathcal{C}}.
\end{equation*}
\end{proposition}
\begin{proof}
Because the conjugacy constraint is required, $s\in S_{\mathcal{A}}$ if and only if $C_s \subseteq S_{\mathcal{A}}$
which is equivalent to $\mathcal{M}_s \subseteq \mathcal{A}(\mathcal{C})$.

If $\mathcal{M}_s \subseteq \mathcal{A}(\mathcal{C})$, from (\ref{eq:derivative}) and Definition \ref{def:derivative ascendant}, we have $\mathcal{D}(\mathcal{M}_s) \subseteq \mathcal{D}\Big(\mathcal{A}(\mathcal{C})\Big) \subseteq \mathcal{C}$.
From Lemma \ref{lemma: 1}, the exponent set of $\mathcal{D}(\mathcal{M}_s)$ is $\texttt{cc}\Big(P(s)\Big)$.
Then $\texttt{cc}\Big(P(s)\Big) \subseteq S_\mathcal{C}$.

If $\texttt{cc}\Big(P(s)\Big) \subseteq S_{\mathcal{C}}$, then $\mathcal{D}(\mathcal{M}_s) \subseteq \mathcal{C}$.
From Definition \ref{def:derivative ascendant}, $\mathcal{A}(\mathcal{C})$ is the extended cyclic code with the largest dimension such that $\mathcal{D}\Big(\mathcal{A}(\mathcal{C})\Big) \subseteq \mathcal{C}$. As a result, $\mathcal{M}_s\subseteq \mathcal{A}(\mathcal{C})$.
\end{proof}

Now we investigate the dimension and distance of $\mathcal{A}(\mathcal{C})$. 
Let $k_{\mathcal{A}}$ and $d_{\mathcal{A}}$ denote the dimension and minimum Hamming distance of $\mathcal{A}(\mathcal{C})$, respectively.
We give the following propositions.
The proofs are given in the Appendix.

\begin{figure*}[htbp]
\centering
\includegraphics[width=0.8\textwidth]{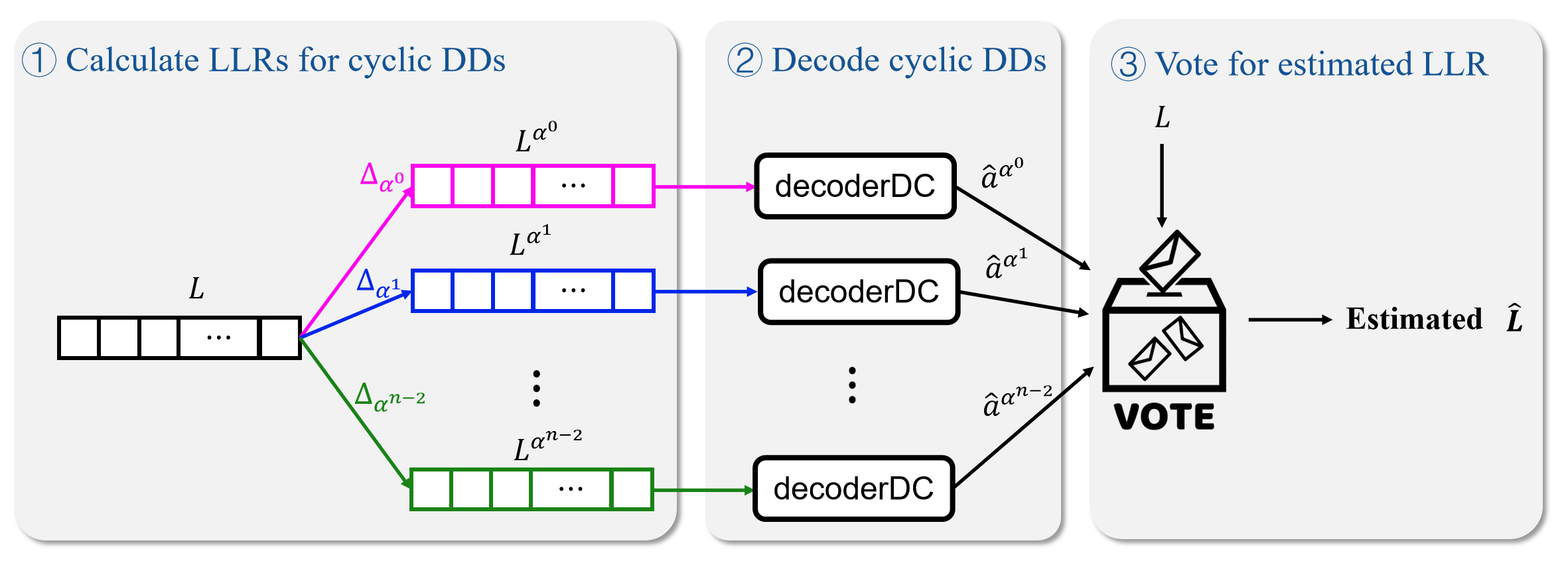}
\caption{Derivative decoding for extended cyclic codes.}
\label{fig:decoding process}
\end{figure*}

\begin{proposition} \label{prop:da dimension}
$k_{\mathcal{A}} \leq \sum_{i=0}^{\texttt{deg}(S_{\mathcal{C}})+1}\binom{m}{i}$.
\end{proposition}

\begin{proposition} \label{prop: da distance}
$d_{\mathcal{A}} \geq d/2$.
\end{proposition}

\section{Derivative Decoding For Extended Binary Cyclic Codes}\label{sec:decoding algorithm}
In this section, we propose a derivative decoding based on the decodings of cyclic DDs. It can efficiently decode the cyclic codes whose cyclic DDs have efficient soft-decision decoding algorithms. 
In particular, we propose to perform the derivative decoding on those codes whose cyclic DDs are extended Euclidean Geometry (EG) codes \cite{kou2001low}\cite[Chap. 8]{lin2001error} which can be efficiently decoded by the \emph{sum-product algorithm} (SPA) \cite{kschischang2001factor,lucas2000iterative}.
In addition, we discuss the cyclic DDs and DAs of RM codes, and their decodings.

\subsection{Algorithm description}

Let $\bm{y} = [y_i, i\in I]$ denote the received vector of transmitting a codeword $\bm{a} = [A(\alpha^i), i\in I]$ of the extended binary cyclic code $\mathcal{C}$ over a binary-input memoryless symmetric (BMS) channel.
Let $W(y|x)$ denote the probability that $y$ is output by the channel when $x$ is input to the channel.
The LLR vector of the channel output $\bm{y}$ is $\bm{L}=[L_i: i \in I]$, where $L_i$ is given by
\begin{equation}\label{eq:llr}
    L_i = \ln( \frac{W(y_i|0)}{W(y_i|1)} ).
\end{equation}

The algorithm takes $\bm{L}$ as an input and runs in an iterative manner. 
Let $B$ denote a collection of directions, i.e. a subset of $\mathbb{F}^*_{2^m}$, where $\mathbb{F}^*_{2^m}$ consists of all the nonzero elements in $\mathbb{F}_{2^m}$. As shown in Fig. \ref{fig:decoding process}, each iteration has three steps: 1) calculate LLR vectors of cyclic DDs for all $\beta \in B$; 2) decode all the cyclic DDs; 3) vote for the estimated codeword from the decoded descendant codewords.

% \addtolength{\topmargin}{0.01in}
\begin{algorithm}
\caption{Derivative Decoding Based on Cyclic Derivative Descendants}\label{alg:DD}
\textbf{Input:} The LLR vector $\bm{L}$; the maximum iteration number $N_{\emph{max}}$; a collection of directions $B$; the parity check matrix $\mathbf{H}$

\textbf{Output:} The decoded codeword: $\hat{\bm{a}}$

\begin{algorithmic}[1]
\For {$t=1,2,\dots,N_{\emph{max}}$} 

\For {$\beta \in B$}

\State $\bm{L}^{\beta} \gets \texttt{derivativeLLR}(\bm{L}, \beta)$

\State $\hat{\bm{a}}^{\beta} \gets \texttt{decoderDD}(\bm{L}^{\beta})$

\State $\widetilde{\bm{L}}^{\beta} \gets \texttt{getVote}(\bm{L}, \hat{\bm{a}}^{\beta}, \beta)$

\EndFor

\State $\bm{L} \gets \frac{1}{|B|}\sum_{\beta\in \mathbb{F}_{2^m}^*}\widetilde{\bm{L}}^{\beta}$
\Comment{Here, $\sum$ denotes the component-wise summation}

\State $\hat{a}_i \gets \mathbbm{1}[L_i<0]$ for all $i\in I$

\If {$\mathbf{H}\bm{a}^{\text{T}} = \bm{0}$}

\State \textbf{Break}

\EndIf

\EndFor

\State \textbf{return} $\hat{\bm{a}}$
\end{algorithmic}

\end{algorithm}

1) The LLR vector associated with the cyclic DD in the direction $\beta$ is defined as 
\begin{equation}\label{eq:llr dd}
    \bm{L}^{\beta}  \triangleq [L^{\beta}_i,i \in I],
\end{equation}
where $L_i^\beta$ is the LLR value associated with $\Delta_{\beta}A(\alpha^i) = A(\alpha^i+\beta) - A(\alpha^i)$.
We calculate $L^{\beta}_i$ as 
\begin{equation}\label{eq:llr dd bit}
    L^{\beta}_i = 2 \tanh ^{-1}\Big(\tanh(\frac{L_i}{2})\tanh(\frac{L_{j}}{2})\Big),
\end{equation}
where $j$ satisfies $\alpha^{j}=\alpha^i+\beta$.
We denote the procedure of calculating $\bm{L}^{\beta}$ with the input $\bm{L}$ and $\beta$ by $\bm{L}^{\beta} = \texttt{derivativeLLR}(\bm{L}\text{, }\beta)$.

2) According to Theorem 1, we can use the same decoder, denoted by $\texttt{decoderDD}$, to decode all the cyclic DDs. The decoding result of $\bm{L}^{\beta}$ is given by $\hat{\bm{a}}^{\beta} = \texttt{decoderDD}(\bm{L}^{\beta})$.

3) The final step is to use a soft-voting scheme to obtain a new LLR vector $\hat{\bm{L}}$.
From (\ref{eq:llr dd bit}), 
the ``soft vote'' from the estimate $\hat{a}_i^{\beta}$ to $L_i$ is $\widetilde{L}_i^{\beta} \triangleq (1-2\hat{a}^{\beta}_i)L_{j}$.
For the direction $\beta$, the ``soft vote'' from $\hat{\bm{a}}^{\beta}$ to $\bm{L}$ is given by
\begin{equation*}
\widetilde{\bm{L}}^{\beta} = \texttt{getVote}(\bm{L}, \hat{\bm{a}}^{\beta}, \beta)
\triangleq [(1-2\hat{a}^{\beta}_i)L_{j}, i\in I].
\end{equation*}
Here we have used the natural embedding of $\mathbb{F}_2$ in $\mathbb{R}$ for the interpretation of $\hat{a}^{\beta}_i$ in the above equation.
Update $\bm{L}$ as the average of all the ``soft votes'' from different directions
$$\bm{L} = \frac{1}{|B|}\sum_{\beta\in \mathbb{F}_{2^m}^*}\widetilde{\bm{L}}^{\beta}.$$
Here, $\sum$ denotes the component-wise summation.

Once we update $\bm{L}$, we take $\hat{\bm{a}}=[\hat{a}_i, i\in I]$ where  $\hat{a}_i = \mathbbm{1}[L_i<0]$.
If $\hat{\bm{a}}$ is a codeword in $\mathcal{C}$, i.e. $\mathbf{H}\hat{\bm{a}}^{\text{T}}=\bm{0}$ where $\mathbf{H}$ is the partiry-check matrix of $\mathcal{C}$ and $\hat{\bm{a}}^\text{T}$ denotes the transpose of $\hat{\bm{a}}$, we end the iteration and output $\hat{\bm{a}}$.
Otherwise, proceed into the next iteration unless obtaining a codeword or reaching a maximal iteration number $N_{\emph{max}}$.
The pseudo code of the above procedure is shown in Algorithm \ref{alg:DD}.

\begin{remark}\label{remark_1}
The proposed decoding works for cyclic codes of length of $2^m -1$ as well. 
Set $L_{\infty}$=0. Then we can decode them as their extended cyclic codes.
\end{remark}

\begin{remark}
In Algorithm \ref{alg:DD}, the functions $\texttt{derivativeLLR}$, $\texttt{decoderDD}$ and $\texttt{getVote}$ can be implemented separately for each direction.
As a result, the proposed algorithm can be implemented in parallel.
\end{remark}

\subsection{Computational complexity}

We analyze the computational complexity of the proposed algorithm per iteration according to Algorithm \ref{alg:DD}. 
Denote the code length by $n$.
It takes $4n$ floating point operations to perform $\texttt{derivativeLLR}$. 
In each iteration, the decoder performs $\texttt{derivativeLLR}$ and $\texttt{decoderDD}$ $|B|$ times. 
At the end of each iteration, it needs $|B|n$ floating point operations for calculating the average of all the ``soft votes''. 
Denote the number of floating point operation of $\texttt{decoderDD}$ by $\Omega$.
We get the following proposition.
\begin{proposition}\label{prop:complexity}
The derivative decoding consumes $5|B|n+|B|\Omega$ floating point operations per iteration.
\end{proposition}

\subsection{Derivative decoding for eBCH codes based on SPA}
Consider the $(64,24)$ eBCH code and the $(64,45)$ eBCH code.
The corresponding generator polynomials in hexadecimal form are \emph{0xF69AC20921} and \emph{0x782CF}, respectively.
From (\ref{eq:exponent set}),
the corresponding representative sets are $S_1=\{0, 1, 3, 5, 9, 21\}$ with $\texttt{deg}(S_1)=3$ and $S_2=\{0, 1, 3, 5, 7, 9, 11, 13, 21, 27\}$ with $\texttt{deg}(S_2)=4$.
According to Theorem \ref{thm:DSN of DD}, 
the representative sets associated with their cyclic DDs are $\{0, 1, 5\}$ and $\{0, 1, 3, 5, 9, 11, 13\}$.
The corresponding codes are the $(64, 13)$ extended EG code and a $(64, 34)$ extended cyclic code. Please note that the $(64, 34)$ extended cyclic code is a subcode of the $(64, 37)$ extended EG code. 

Consider \emph{additive white Gaussian noise} (AWGN) channels. We decode the two eBCH codes by our proposed derivative decoding algorithm.
The decoder for their cyclic DDs is the SPA decoder. 
The parity-check matrix used for decoding the $(64, 13)$ extended EG code is a $336\times 64$ matrix with row weight $4$ and column weight $21$,
and the one used for decoding the $(64, 34)$ extended cyclic code is a $72\times 64$ matrix with row weight $8$ and column weight $17$.
We set the maximum iteration numbers for SPA and DD as $N_{\emph{SPA,max}}=20$ and $N_{\emph{DD,max}}=3$, respectively.
We perform derivative decoding in all the $63$ directions and denote the procedure by DD$(63)$-SPA. 
In addition, we perform derivative decoding in $16$ directions at random and denote the procedure by DD$(16)$-SPA.
We compare with the performance of decoding BCH codes using the Berlekamp-Massey (BM) algorithm \cite{berlekamp2015algebraic}\cite{massey1969shift}.

The simulation results are shown in Fig. \ref{fig:bch2eg} and the performance of the MLD is also provided.
We see that at the block error ratio (BLER) of $10^{-4}$, DD$(63)$-SPA for the $(64,24)$ eBCH code and the $(64, 45)$ eBCH code outperforms BM for the corresponding BCH codes about $2.9$ dB and $1.9$ dB, respectively.  Moreover, the gaps between the MLD and DD$(63)$-SPA are $0.5$ dB and $0.3$ dB in, respectively. In addition, for decoding the $(64, 45)$ eBCH code, DD$(16)$-SPA only performs about $0.2$ dB away from DD$(63)$-SPA at the BLER of $10^{-4}$.

\begin{figure*}[htbp]
\centering
\subfigure[$(64, 24)$ eBCH]{%
\begin{minipage}[t]{0.45\linewidth}
\flushleft 
\label{fig:data64_24}%
\includegraphics[width=1.0\textwidth]{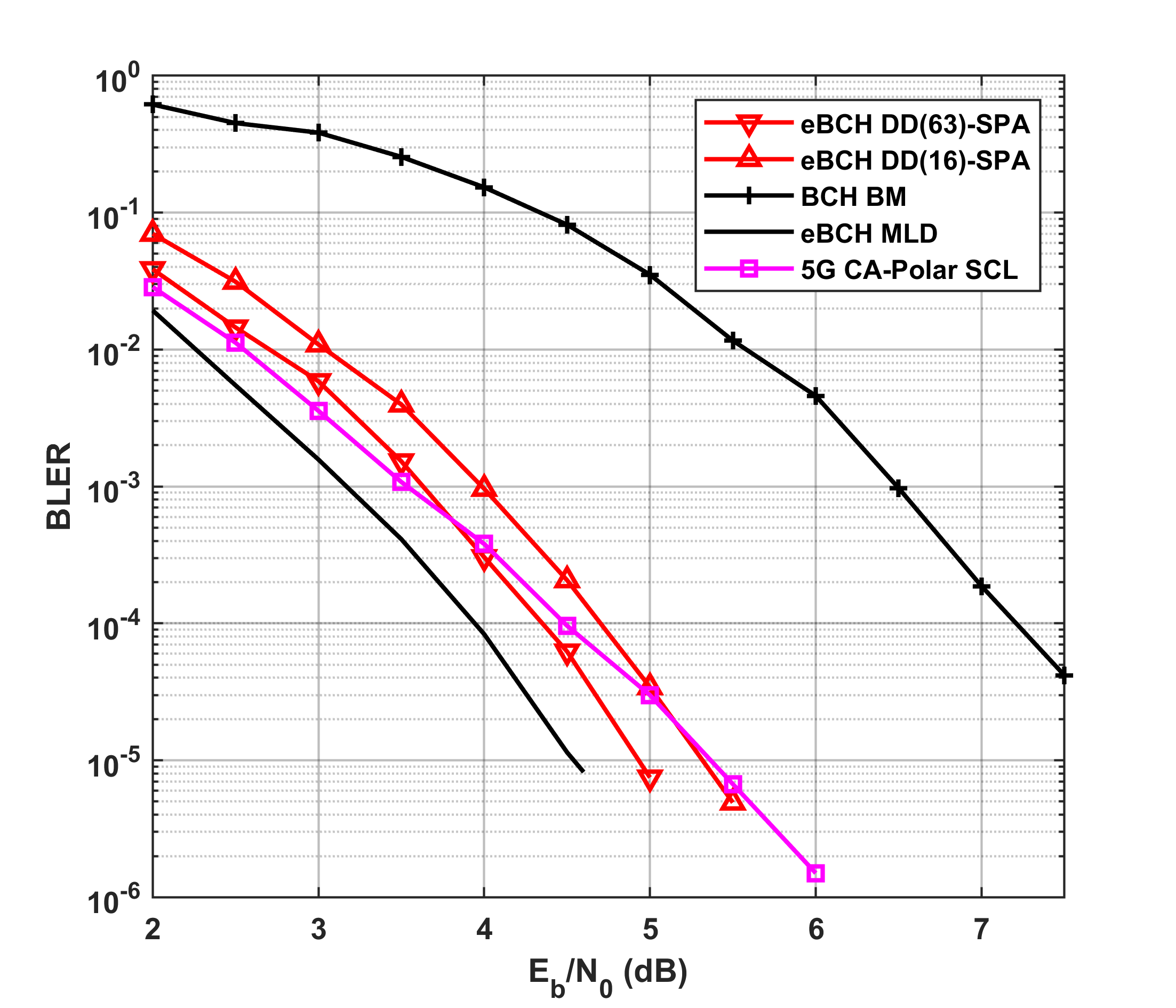}
\end{minipage}
}
\subfigure[$(64, 45)$ eBCH]{%
\begin{minipage}[t]{0.45\linewidth}
\flushleft 
\label{fig:data64_45}%
\includegraphics[width=1.0\textwidth]{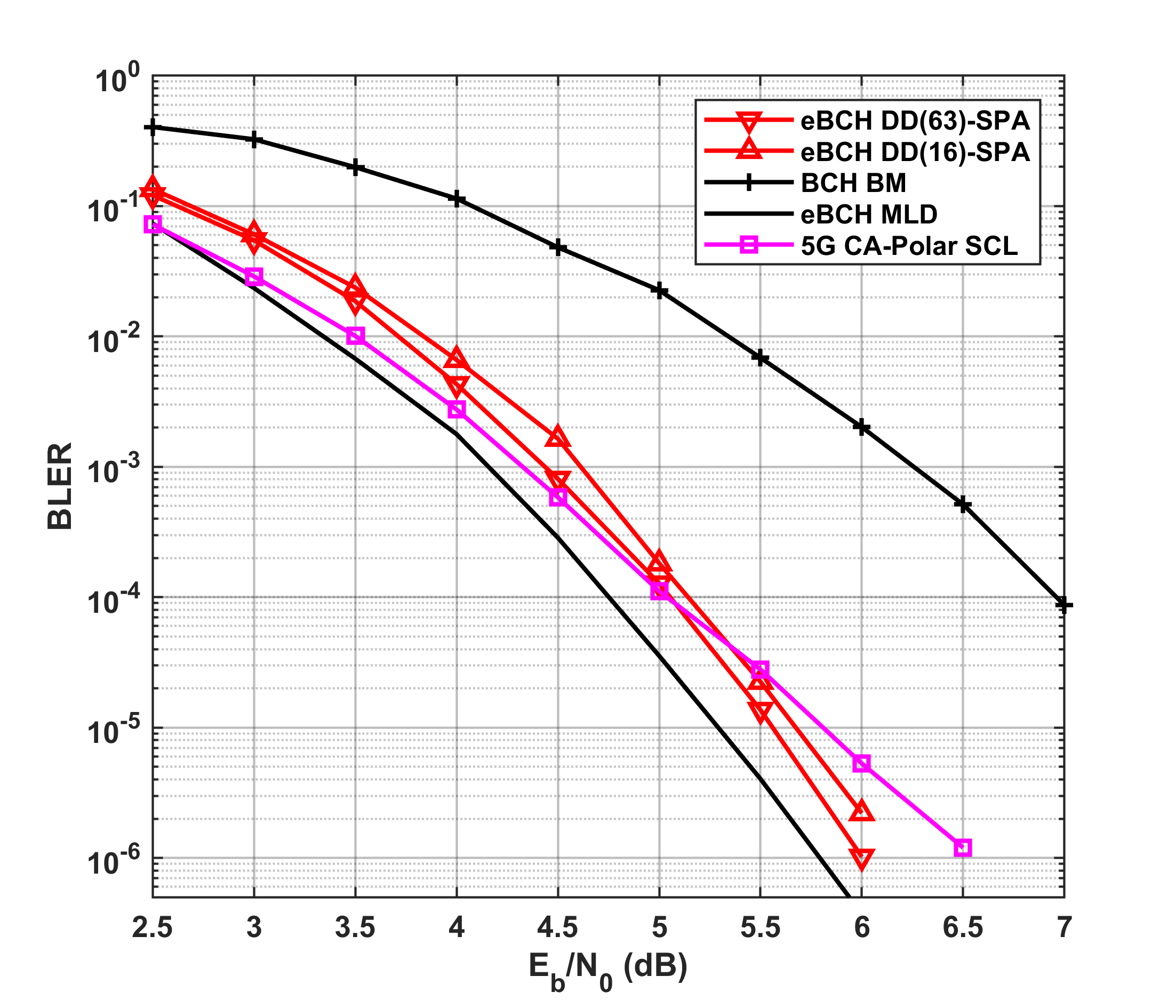}
\end{minipage}
}
\caption{Performance of decoding eBCH codes using DD-SPA, decoding BCH codes using BM algorithm and decoding CA-polar codes using SCL decoder. 
The list size of the SCL decoder is $32$ and the CRC length is $6$.}
\label{fig:bch2eg}
\end{figure*}

Besides, we compare with the performance of decoding 5G CA-polar codes \cite{arikan2009channel}\cite{TS38212} with the same length and dimension.
The decoder used for the CA-polar codes is the Successive Cancellation List (SCL) decoder\cite{tal2015list} with list size $32$. And the CRC length is set to $6$. At the BLER of $10^{-5}$, DD($63$)-SPA for the $(64,24)$ eBCH code and the $(64, 45)$ eBCH code outperforms SCL for the CA-polar codes with the same length and dimension about $0.4$ dB and $0.2$ dB, respectively.

We discuss the computational complexity of decoding the $(64, 45)$ eBCH code using DD$(63)$-SPA. Denote the floating point operation cost of the SPA in one iteration by $\Omega_{\emph{SPA}}$. From Proposition \ref{prop:complexity}, the number of the floating point operations is $320*63+63N_{\emph{SPA}}\Omega_{\emph{SPA}}$ per iteration where $\Omega_{\emph{SPA}}$ is the average iteration number of the SPA. 
At the $E_b/N_0$ of $5.0$ dB, the average iteration number of the SPA is $1.16$ and the average iteration number of DD$(63)$-SPA is $1.00$.
As a result, the cost of DD$(63)$-SPA is around $20160+73\Omega_{\emph{SPA}}$ floating point operations.
Besides, the average iteration number of the SPA in DD$(16)$-SPA and the average iteration number of DD$(16)$-SPA is also $1.16$ and $1.00$, respectively. 
Then the cost of DD$(16)$-SPA is around $5120 +  19\Omega_{\emph{SPA}}$ floating point operations.

\subsection{Cyclic DDs of RM codes and their decoding}

Consider the RM code of length $2^m$ and order $r$ denoted by RM$(r,m)$.
The zero set of the generator polynomial associated with RM$(r,m)$ is $\{\alpha^j:0<\texttt{wt}(\overline{j})<m-r\}$ \cite[Chap. 6]{blahut2003algebraic}.
According to (\ref{eq:exponent set}), the exponent set of RM$(r,m)$ is $\{j:0\leq\texttt{wt}(\overline{j})\leq r\}$.
According to Theorem \ref{thm:DSN of DD}, we conclude that the cyclic DD of RM$(r, m)$ is RM$(r-1, m)$. Note that the dimension and minimum Hamming distance of RM$(r, m)$ are $\sum_{i=0}^{r}\binom{m}{i}$ and $2^{m-r}$, respectively. The equalities in Proposition \ref{prop:dd dimension} and Proposition \ref{prop:distance of dd} hold for RM codes and their cyclic DDs. Similarly, according to Proposition \ref{prop:ascendant DS}, we conclude that the cyclic DA of RM$(r, m)$ is RM$(r+1, m)$. The equalities in Proposition \ref{prop:da dimension} and Proposition \ref{prop: da distance} hold for the RM codes and their cyclic DAs.

As a result, we can decode RM$(r,m)$ codes using derivative decoding based on the decodings of RM$(r-1,m)$ codes.
Let $T$ denote a subset of $\mathbb{F}_{2^m}$ such that $(\beta+T)\cup T= \mathbb{F}_{2^m}$\footnote{For a subset $T$ of $\mathbb{F}_{2^m}$, $\beta + T$ denote the set $\{\beta+\alpha^i : \alpha^i \in T\}$.}.
From (\ref{eq:derivative}), we have $[\Delta_{\beta}A(\alpha^{i}), \alpha^i\in T] = [\Delta_{\beta}A(\alpha^{i}), \alpha^i\in \beta+T]$.
Considering the $|\bm{u}|\bm{u}+\bm{v}|$-construction and the automorphism groups of RM codes\cite[Chap. 13]{macwilliams1977theory}, we conclude that $[\Delta_{\beta}A(\alpha^{i}), \alpha^i\in T]$ is a codeword of RM$(r-1,m-1)$. The proof is given in the Appendix.
From (\ref{eq:llr dd bit}), we have $[L^{\beta}_{i}, \alpha^i\in T] = [L^{\beta}_{i}, \alpha^i\in \beta+T]$, 
where $[L^{\beta}_{i}, \alpha^i\in T]$ is the LLR vector associated with $[\Delta_{\beta}A(\alpha^{i}), \alpha^i\in T]$. 
In other words, our derivative decoding derived from the MS polynomials can carry on based on the decodings of RM$(r-1,m-1)$ codes as the state-of-the-art projection decodings \cite{ye2020recursive,lian2020decoding} derived from the $m$-variate polynomials, and obtain the same performance.

\subsection{Decoding cyclic derivative ascendants of EG codes}

If an extended cyclic code can be efficiently soft-decision decoded, then its cyclic DA can be soft-decision decoded by the derivative decoding algorithm.
Consider the $(256, 175)$ extended EG code with minimum Hamming distance 18 which can be efficiently decoded by the SPA decoder. 
The corresponding generator polynomial is \emph{0x11377F7700FA55335BA55}.
The representative set of its exponent set is $S=\{0,1,3,5,7,9,11,13,17,19,21,$ $23,25,27, 29,37,39 ,43,51,53,55,59,85,87,119\}$ with $\texttt{deg}(S)=6$.
According to Proposition \ref{prop:ascendant DS}, we can construct its cyclic DA. It is an extended cyclic code of length $256$ and dimension $191 \leq \sum_{i=0}^{7}\binom{8}{i}$ with distance $d \geq 18/2=9$, according to Proposition \ref{prop:da dimension} and Proposition \ref{prop: da distance}. In fact, this code, denoted by DA$(256, 191)$ has minimum Hamming distance at least $16$ according to the BCH bound \cite{blahut2003algebraic}.
The corresponding generator polynomial of DA$(256, 191)$ is \emph{0x19ACCC1AE68A0CEFF}.

In Fig. \ref{fig:data256_191}, we provide the simulation result of decoding DA$(256, 191)$ using derivative decoding based on SPA with all the directions in $\mathbb{F}^*_{2^m}$, denoted by DD($255$)-SPA.
The parity-check matrix used for decoding the $(256, 175)$ extended EG code is a $272\times256$ matrix with row weight $16$ and column weight $17$.
The maximum iteration numbers for derivative decoding and SPA are set to $N_{\emph{DD,max}}=4$ and $N_{\emph{SPA,max}}=20$, respectively.
The performance of the MLD is also provided.
We see that at the BLER of $10^{-4}$, the gap between the MLD and DD-SPA is about $0.9$ dB.

\begin{figure}[htbp]
\centering
\includegraphics[width=0.5\textwidth]{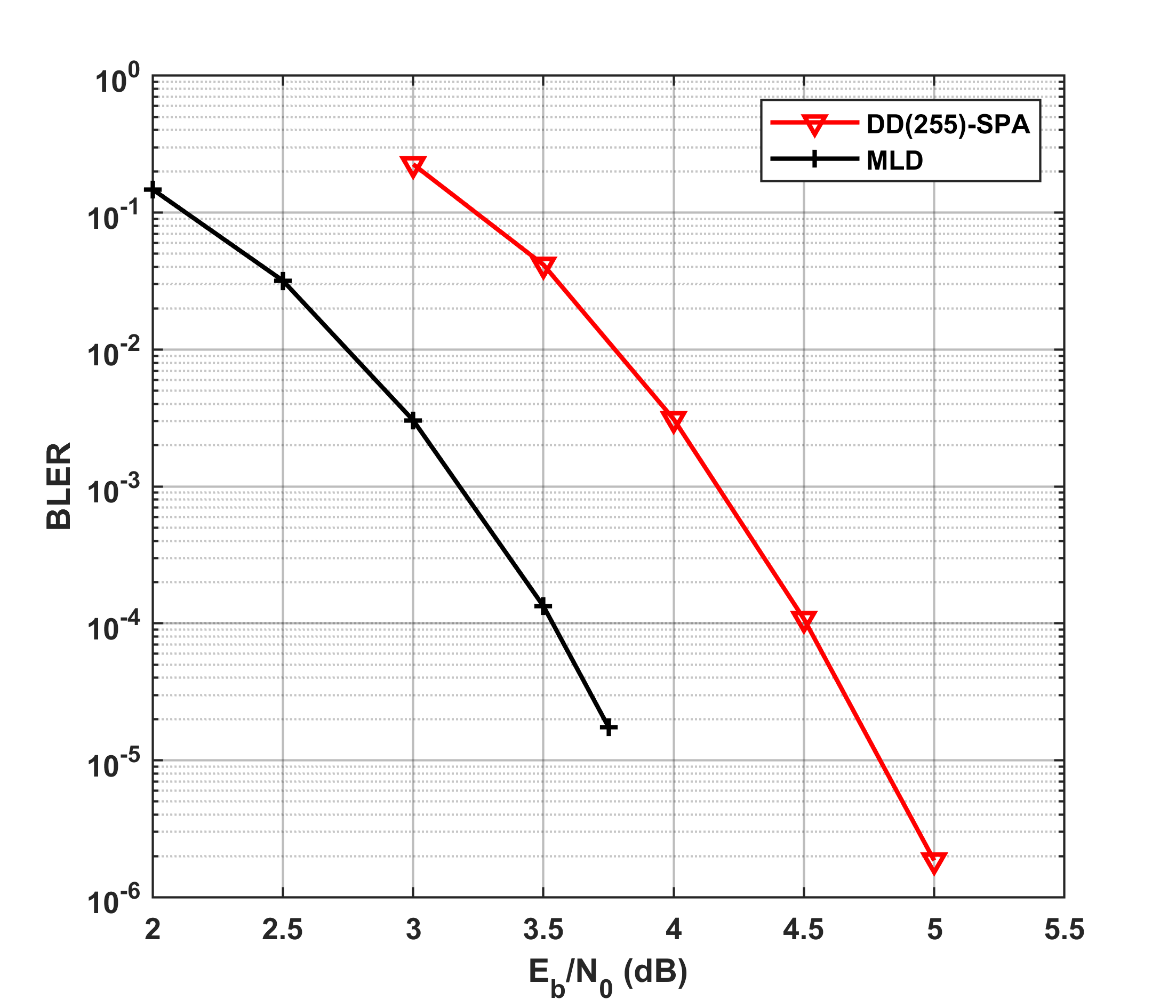}
\caption{Performance of decoding DA$(256, 191)$ using DD$(255)$-SPA.}
\label{fig:data256_191}
\end{figure}

\section{Minimal Derivative Descendants and Derivative Decoding}
\label{sec:MDD}

This section investigates the minimal subspace which contains all the derivatives of an extended cyclic code in one direction. We prove that all these subspaces, denoted as the minimal DDs, are equivalent. Similarly, we can decode an extended cyclic code based on the decodings of its minimal DDs. Simulation results show that the derivative decoding based on the OSD with order-$1$ can outperform the OSD with higher order. 
In the following, we denote the OSD with order-$l$ by OSD$(l)$.

\subsection{Minimal derivative descendants of extended cyclic codes}
\begin{definition}\label{def:derivative descendant}
Consider an extended cyclic code $\mathcal{C}$ and the direction $\beta$. We denote the minimal subspace which contains the derivatives of all the codewords of $\mathcal{C}$ in $\beta$ as the \emph{minimal derivative descendant} in $\beta$
\begin{equation*}
\mathcal{D}_{\beta}(\mathcal{C})=\{[\Delta_{\beta}A(\alpha^i), i \in I]: A(z)\in \mathcal{C} \}.
\end{equation*}
\end{definition}

Two codes are said to be \emph{equivalent} \cite[Chap. 1]{macwilliams1977theory} if there is a permutation of the coordinates together with permutations of the coordinate values for each of the coordinates, that map the codewords of one code into those of the other. The following theorem proves that the minimal DDs in different directions are equivalent.

For an integer $b$, we make the agreement $\infty + b = \infty \text{ mod }n$. For a codeword $\bm{a} = [a_i, i\in I] \in \mathcal{C}$, we define the $b$-cyclic shift of $\bm{a}$ as $\bm{a}^{(b)}\triangleq[a_{i+b}, i\in I]$. 

\begin{theorem}\label{thm:equivalent}
The minimial DDs of an extended cyclic code $\mathcal{C}$ in different directions are equivalent.
\end{theorem}
\begin{proof}
Let $\beta_1$ and $\beta_2$ be powers of $\alpha$, i.e. $\beta_1=\alpha^{b_1}$ and $\beta_2=\alpha^{b_2}$.
For any codeword $\Delta_{\beta_1}A(z) \in \mathcal{D}_{\beta_1}(\mathcal{C})$, there is $\Delta_{\beta_2}A(\beta_2^{-1}\beta_1z) \in \mathcal{D}_{\beta_2}(\mathcal{C})$ such that
\begin{equation}\label{eq:cyclicshiftofderivative}
\begin{aligned}
\Delta_{\beta_2}A(\beta_2^{-1}\beta_1z) &= A\Big(\beta_2^{-1}\beta_1(z+\beta_2)\Big) - A(\beta_2^{-1}\beta_1z)\\
& = A(\alpha^{b_1-b_2}z+\beta_1) - A(\alpha^{b_1-b_2}z),
\end{aligned}
\end{equation}
is the $(b_1-b_2)$-cyclic shift of $\Delta_{\beta_1}A(z)$.
Please note that this is true for any pair of $\beta_1$ and $\beta_2$.
As a result, the minimal DDs of $\mathcal{C}$ in all the directions are equivalent.
\end{proof}
The above theorem shows that the cyclic shift of a codeword in one minimal DD is a codeword in another minimal DD.
We obtain the following corollary immediatly from (\ref{eq:cyclicshiftofderivative}) by setting $b_1=b$ and $b_2=0$.

\begin{corollary}\label{coro:coro_1}
For any codeword $A(z)\in \mathcal{C}$, the $b$-cyclic shift of $\Delta_{\alpha^b}A(z)$ is equal to $\Delta_{1}A(\alpha^b z)$.
\end{corollary}

According to Definitions \ref{def:cyclic derivative descendant} and \ref{def:derivative descendant}, the minimal DDs are the subcodes of the corresponding cyclic DDs. 
In the next proposition, we show that the cyclic DD of an extended cyclic code $\mathcal{C}$ is the summation of all the minimal DDs of $\mathcal{C}$.
\begin{proposition}\label{prop:summation of mDD is cDD}
For an extended cyclic code $\mathcal{C}$, the summation of all its minimal DDs is its cyclic DD, i.e., 
\begin{equation}
\mathcal{D}(\mathcal{C}) = \sum_{\beta\in \mathbb{F}_{2^m}^*} D_{\beta}(\mathcal{C}).
\end{equation}
\end{proposition}
\begin{proof}
We are going to prove that $\sum_{\beta\in \mathbb{F}_{2^m}^*} D_{\beta}(\mathcal{C})$ is the smallest extended cyclic code containing all the minimal DDs of $\mathcal{C}$.
Obviously, it is the smallest subspace of $\mathbb{F}_2^n$ containing all the minimal DDs of $\mathcal{C}$.
We only need to prove that $\sum_{\beta\in \mathbb{F}_{2^m}^*} D_{\beta}(\mathcal{C})$ is an extended cyclic code.

For any codeword $\Delta_{\beta}A(z) = A(z+\beta)-A(z)$ in $\sum_{\beta\in \mathbb{F}_{2^m}^*} D_{\beta}(\mathcal{C})$, its cyclic shift
\begin{equation}
\begin{aligned}
A(\alpha z+\beta)-A(\alpha z) & = A\Big(\alpha(z+\alpha^{-1}\beta)\Big)-A(\alpha z) \\
& = \Delta_{\alpha^{-1}\beta}A(\alpha z),
\end{aligned}
\end{equation}
is also a codeword in $\sum_{\beta\in \mathbb{F}_{2^m}^*} D_{\beta}(\mathcal{C})$.
As a result, $\sum_{\beta\in \mathbb{F}_{2^m}^*} D_{\beta}(\mathcal{C})$ is an extended cyclic code.
\end{proof}

In general, it is hard to determine the minimal Hamming distance of an arbitrary linear code.
With Proposition \ref{prop:summation of mDD is cDD}, we can tell that the minimal Hamming distance of a minimal DD of $\mathcal{C}$ is lower bounded by the minimal Hamming distance of the cyclic DD of $\mathcal{C}$, which can be lower bounded by the BCH bound \cite{bose1960further,bose1960class}.

\begin{example}
Continuation of Example \ref{ex:ex_2}. 
The generator matrix $\bm{G}$ of the $(16, 7)$ extended cyclic code is
\begin{equation*}
\begin{aligned}
\left[ 
\begin{array}{cccccccccccccccc}
1 & 1 & 0 & 0 & 0 & 1 & 0 & 1 & 1 & 1 & 0 & 0 & 0 & 0 & 0 & 0\\
1 & 0 & 1 & 0 & 0 & 0 & 1 & 0 & 1 & 1 & 1 & 0 & 0 & 0 & 0 & 0\\
1 & 0 & 0 & 1 & 0 & 0 & 0 & 1 & 0 & 1 & 1 & 1 & 0 & 0 & 0 & 0\\
1 & 0 & 0 & 0 & 1 & 0 & 0 & 0 & 1 & 0 & 1 & 1 & 1 & 0 & 0 & 0\\
1 & 0 & 0 & 0 & 0 & 1 & 0 & 0 & 0 & 1 & 0 & 1 & 1 & 1 & 0 & 0\\
1 & 0 & 0 & 0 & 0 & 0 & 1 & 0 & 0 & 0 & 1 & 0 & 1 & 1 & 1 & 0\\
1 & 0 & 0 & 0 & 0 & 0 & 0 & 1 & 0 & 0 & 0 & 1 & 0 & 1 & 1 & 1\\
\end{array}
\right]. \\
\end{aligned}
\end{equation*}
The columns are indexed by $0$, $\alpha^0$, $\alpha^1$, ..., $\alpha^{14}$.
Calculate the derivatives of the rows of $\bm{G}$ in $\alpha^0$  and obtain the following matrix 
\begin{equation*}
\begin{aligned}
 \left[ 
\begin{array}{cccccccccccccccc}
     0&     0&     1&     1&     0&     1&     0&     1&     1&     1&     1&     0&     0&     0&     1&     0 \\
     1&     1&     1&     1&     0&     1&     1&     0&     0&     1&     0&     1&     0&     0&     0&     0 \\
     1&     1&     0&     0&     0&     0&     1&     1&     1&     0&     1&     1&     0&     0&     1&     0 \\
     1&     1&     0&     0&     1&     0&     1&     0&     0&     0&     0&     1&     1&     1&     0&     1 \\
     1&     1&     1&     1&     0&     1&     1&     0&     0&     1&     0&     1&     0&     0&     0&     0 \\
     1&     1&     0&     0&     0&     0&     1&     1&     1&     0&     1&     1&     0&     0&     1&     0 \\
     1&     1&     0&     0&     1&     0&     1&     0&     0&     0&     0&     1&     1&     1&     0&     1 \\
\end{array}
\right]. \\
\end{aligned}
\end{equation*}
The minimal DD $\mathcal{D}_{1}(\mathcal{C})$ of $\mathcal{C}$ is spanned by the rows of the above matrix.
We can perform Gauss elimination on the above matrix and obtain the generator matrix of $\mathcal{D}_{1}(\mathcal{C})$,
\begin{equation*}
\begin{aligned}
\left[ 
\begin{array}{cccccccccccccccc}
     1&     1&     0&     0&     0&     0&     1&     1&     1&     0&     1&     1&     0&     0&     1&     0\\
     0&     0&     1&     1&     0&     1&     0&     1&     1&     1&     1&     0&     0&     0&     1&     0\\
     0&     0&     0&     0&     1&     0&     0&     1&     1&     0&     1&     0&     1&     1&     1&     1\\
\end{array}
\right]. \\
\end{aligned}
\end{equation*}
It shows that $\mathcal{D}_{1}(\mathcal{C})$ is a $(16, 3)$ code.
From Proposition \ref{prop:summation of mDD is cDD} and Example \ref{ex:ex_2}, $\mathcal{D}_{1}(\mathcal{C})$ is a subcode of the $(16, 5)$ Hadamard code whose minimal Hamming distance is $8$. Thus, its minimal Hamming distance is at least $8$. In fact, its minimal Hamming distance is exactly $8$.

Consider the codeword $\bm{a} = $ $[$$1$ $0$ $1$ $0$ $0$ $0$ $1$ $0$ $1$ $1$ $1$ $0$ $0$ $0$ $0$ $0$$]$ in $\mathcal{C}$ and its cyclic shift $\bm{a}^{(1)} = $ $[$$1$ $1$ $0$ $0$ $0$ $1$ $0$ $1$ $1$ $1$ $0$ $0$ $0$ $0$ $0$ $0$$]$. The derivative of $\bm{a}$ in $\alpha$ is $[$$0$ $0$ $0$ $1$ $1$ $0$ $1$ $0$ $1$ $1$ $1$ $1$ $0$ $0$ $0$ $1$$]$.
The derivative of $\bm{a}$ in $\alpha^0$ is  $[$$0$ $0$ $1$ $1$ $0$ $1$ $0$ $1$ $1$ $1$ $1$ $0$ $0$ $0$ $1$ $0$$]$ which is a cyclic shift of the derivative of $\bm{a}$ in $\alpha$.
\end{example}

\subsection{Decoding based on decodings of minimal DDs}
According to Theorem \ref{thm:equivalent} and Corollary \ref{coro:coro_1}, we can perform derivative decoding on $\mathcal{C}$ based on the decodings of $\mathcal{D}_1(\mathcal{C})$ with cyclic shiftings.

Consider transmitting a codeword $\bm{a} = [A(\alpha^i), i\in I]$ over a BMS channel, and the recieved vector is $\bm{y}$. 
Let $\bm{L}$ denote the corresponding LLR vector. 
For the direction $\beta$, we denote the vector $[\Delta_{\beta}A(\alpha^i), i\in I]$ by $\bm{a}^{\beta}$.
From Corollary \ref{coro:coro_1}, the $b$-cyclic shift of $\bm{a}^{\beta}$, denoted by $\bm{a}^{\beta, (b)}$, is equal to the derivative of $\bm{a}^{(b)}$ in $\alpha^0$.
We denote the $b$-cyclic shift of $\bm{L}$ as $\bm{L}^{(b)}$.
Then we calculate the LLR vector associated with $\bm{a}^{\beta, (b)}$ as
\begin{equation}\label{eq:llr dd 2}
    \bm{L}^{\beta, (b)}  \triangleq [L^{{}\beta, (b)}_i, i \in I],
\end{equation}
where
\begin{equation}\label{eq:llr dd bit 2}
L^{\beta, (b)}_i = 2 \tanh ^{-1}\Big(\tanh(\frac{L^{(b)}_i}{2})\tanh(\frac{L^{(b)}_{j}}{2})\Big),
\end{equation}
where $j$ satisifies $\alpha^j = \alpha^i+1$.
Now we can treat $\bm{L}^{\beta, (b)}$ as a LLR vector associated with an codeword in $\mathcal{D}_1(\mathcal{C})$.
Denote a soft-decision decoder for $\mathcal{D}_1(\mathcal{C})$ by $\texttt{decoderDD}$.
The estimate of $\bm{a}^{\beta, (b)}$ is given by ${\hat{\bm{a}}^{\beta, (b)}} = \texttt{decoderDD}(\bm{L}^{\beta, (b)})$.
According to (\ref{eq:llr dd bit 2}),
the ``soft vote'' for $L_i^{(b)}$ from the direction $\beta$ is $\widetilde{L}_i^{\beta, (b)} = (1-2\hat{a}^{\beta, (b)}_i)L^{(b)}_{j}$.
And the ``soft vote''  for $\bm{L}^{(b)}$ from the direction $\beta$ is given by
\begin{equation*}
\widetilde{\bm{L}}^{\beta, (b)} = \texttt{getVote}(\bm{L}^{(b)}, \bm{a}^{\beta, (b)}, \alpha^0).
\end{equation*}
Cyclicly shift it $b$ places to the right and obtain the ``soft vote'' for $\bm{L}$ from the direction $\beta$, i.e. $\widetilde{\bm{L}}^{\beta}$.
The remaining steps mimic to Algorithm \ref{alg:DD} in  Section \ref{sec:decoding algorithm}.
We provide the pseudo code in Algorithm \ref{alg:DD2}.

\begin{algorithm}
\caption{Derivative Decoding Based on Minimal Derivative Descendants}\label{alg:DD2}
\textbf{Input:} The LLR vector $\bm{L}$; the maximum iteration number $N_{\emph{max}}$; a collection of directions $B$;  the parity check matrix $\mathbf{H}$

\textbf{Output:} The decoded codeword: $\hat{\bm{a}}$

\begin{algorithmic}[1]
\For {$t=1,2,\dots,N_{\emph{max}}$} 

\For {$b = 1, 2, \dots, n$}

\State $\bm{L}^{(b)} \gets \bm{L}^{(b-1)}$ 
\Comment{Take $\bm{L}$ as $\bm{L}^{(0)}$}

\If {$\alpha^b \in B$}

\State $\bm{L}^{\beta, (b)} \gets \texttt{derivativeLLR}(\bm{L}^{(b)}, 1)$

\State $\hat{\bm{a}}^{\beta, (b)} \gets \texttt{decoderDD}(\bm{L}^{\beta, (b)})$

\State $\widetilde{\bm{L}}^{\beta, (b)} \gets \texttt{getVote}(\hat{\bm{a}}^{\beta, (b)}, \bm{L}^{b}, 1)$

\State $\widetilde{\bm{L}}^{\beta} \gets \widetilde{\bm{L}}^{\beta, (b)} $ 
\Comment{Cyclicly shift $\widetilde{\bm{L}}^{\beta, (b)}$ $b$ places to the right}

\EndIf

\EndFor

\State $\bm{L} \gets \frac{1}{|B|}\sum_{\beta\in B}\widetilde{\bm{L}}^{\beta}$
\Comment{Here, $\sum$ denotes the component-wise summation}

\State $\hat{a}_i \gets \mathbbm{1}[L_i<0]$ for all $i\in I$

\If {$\mathbf{H}\bm{a}^{\text{T}} = \bm{0}$}

\State \textbf{Break}

\EndIf

\EndFor

\State \textbf{return} $\hat{\bm{a}}$
\end{algorithmic}

\end{algorithm}

\begin{remark}
The advantage of calculating $\bm{L}^{\beta, (b)}$ rather than $\bm{L}^{\beta}$ is that we can treat $\bm{L}^{\beta, (b)}$ as the LLR vector associated with the direction $\alpha^0$ for all $\beta \in \mathbb{F}_{2^m}$.
It allows us to keep using \texttt{derivative}, \texttt{decoderDD} and \texttt{getVote} for the direction $\alpha^0$.
\end{remark}

\begin{remark}\label{remark4}
For any $A(z) \in \mathcal{C}$, there is $\Delta_{1}A(\alpha^i) = \Delta_{1}A(\alpha^i+1)$.
Denote a subset of $\mathbb{F}_{2^m}$ by $T$ such that $T \cup (1+T) = \mathbb{F}_{2^m}$. 
For any codeword $\Delta_{1}A(z) \in \mathcal{D}_{1}(\mathcal{C})$, there is $[\Delta_{1}A(\alpha^i), \alpha^i \in T] = [\Delta_{1}A(\alpha^i), \alpha^i \in 1+T]$.
Besides, from (\ref{eq:llr dd bit 2}), there is $[L^{\beta, (b)}_{i}, \alpha^i \in 1+T] = [L^{\beta, (b)}_{i}, \alpha^i \in 1+T]$.
It can be used to simplify the decoding for minimal DDs in Algorithm \ref{alg:DD2}.
\end{remark}

\subsection{Derivative decoding based on OSD}
We propose to perform derivative decoding based on OSD.
Suppose $A_1(z)$, $A_2(z)$, ..., $A_{k_{\mathcal{D}}}(z)$ forms a basis of $\mathcal{D}_1(\mathcal{C})$. We can perform OSD based on the generator matrix given by this basis,
\begin{equation*}
\begin{aligned}
\mathbf{G}_{\mathcal{D}} & = 
\left[
\begin{array}{cccc}
A_1(\alpha^{\infty}) & A_1(\alpha^{0}) & ... & A_1(\alpha^{n-1})\\
A_2(\alpha^{\infty}) & A_2(\alpha^{0}) & ... & A_2(\alpha^{n-1})\\
... & ... & ... & ...\\
A_{k_{\mathcal{D}}}(\alpha^{\infty}) & A_{k_{\mathcal{D}}}(\alpha^{0}) & ... & A_{k_{\mathcal{D}}}(\alpha^{n-1})\\ 
\end{array}
\right] \\
& \triangleq 
\begin{array}{cccc}
[\bm{g}_{\alpha^\infty} & \bm{g}_{\alpha^0} & ... & \bm{g}_{\alpha^{n-1}}].
\end{array}
\end{aligned}
\end{equation*}
In fact, from Remark \ref{remark4}, we can implement the OSD decoder with the matrix $[\bm{g}_{\alpha^i}, \alpha^i \in T]$ and only take $[L^{\beta, (b)}_{i}, \alpha^i \in T]$ as the input when decoding $\bm{L}^{\beta, (b)}$ in Algorithm \ref{alg:DD2}.
We denote the derivative decoding based on OSD with by DD-OSD.
In particular, we focus on the derivative decoding based on OSD$(1)$ and denote it by DD-OSD$(1)$.

\begin{table}

\centering
\caption{Code Parameters of Extended BCH Codes and their descendants}
\begin{tabular}{cccccc}
   \toprule
   $n_{\mathcal{C}}$ & $k_{\mathcal{C}}$ & $d_{\mathcal{C}, \textbf{BCH}}$ &$k_{\mathcal{D}}$ & $d_{\mathcal{D}, \textbf{BCH}}$ & $k_{\mathcal{D}_1}$  \\
   \midrule

128 & 36 & 32 & 22 & 48 & 14 \\
256 & 37 & 92 & 25 & 96 & 16 \\
256 & 79 & 56 & 45 & 64 & 31 \\

   \bottomrule
\end{tabular}
\label{tb: eBCH codes}
\end{table}

\begin{figure*}[htbp]
\subfigure[$(128, 36)$ eBCH]{%
\begin{minipage}[t]{0.33\linewidth}
\flushleft 
\label{fig:data128_36}%
\includegraphics[width=1.0\textwidth]{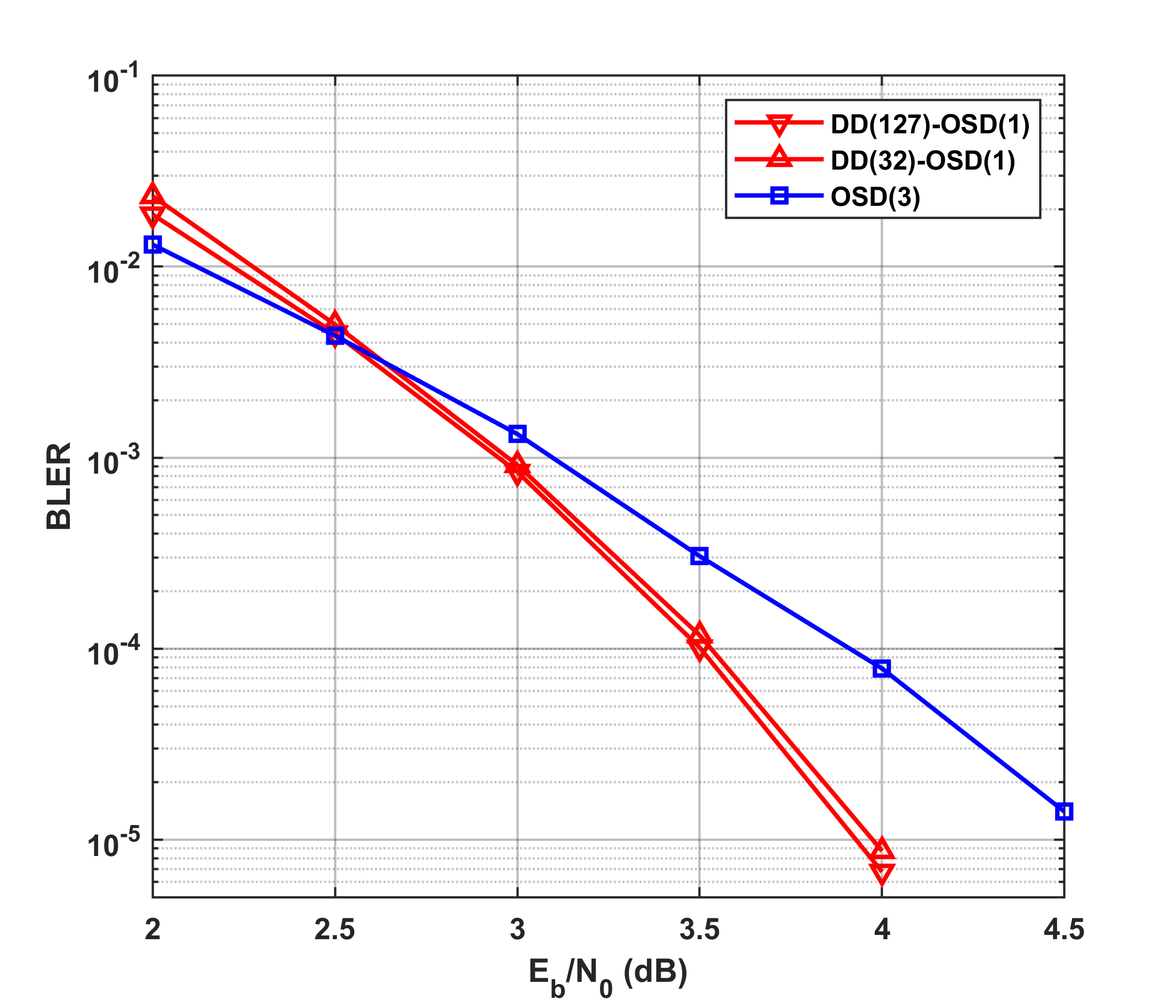}
\end{minipage}
}
\subfigure[$(256, 37)$ eBCH]{%
\begin{minipage}[t]{0.33\linewidth}
\flushleft 
\label{fig:data256_37}%
\includegraphics[width=1.0\textwidth]{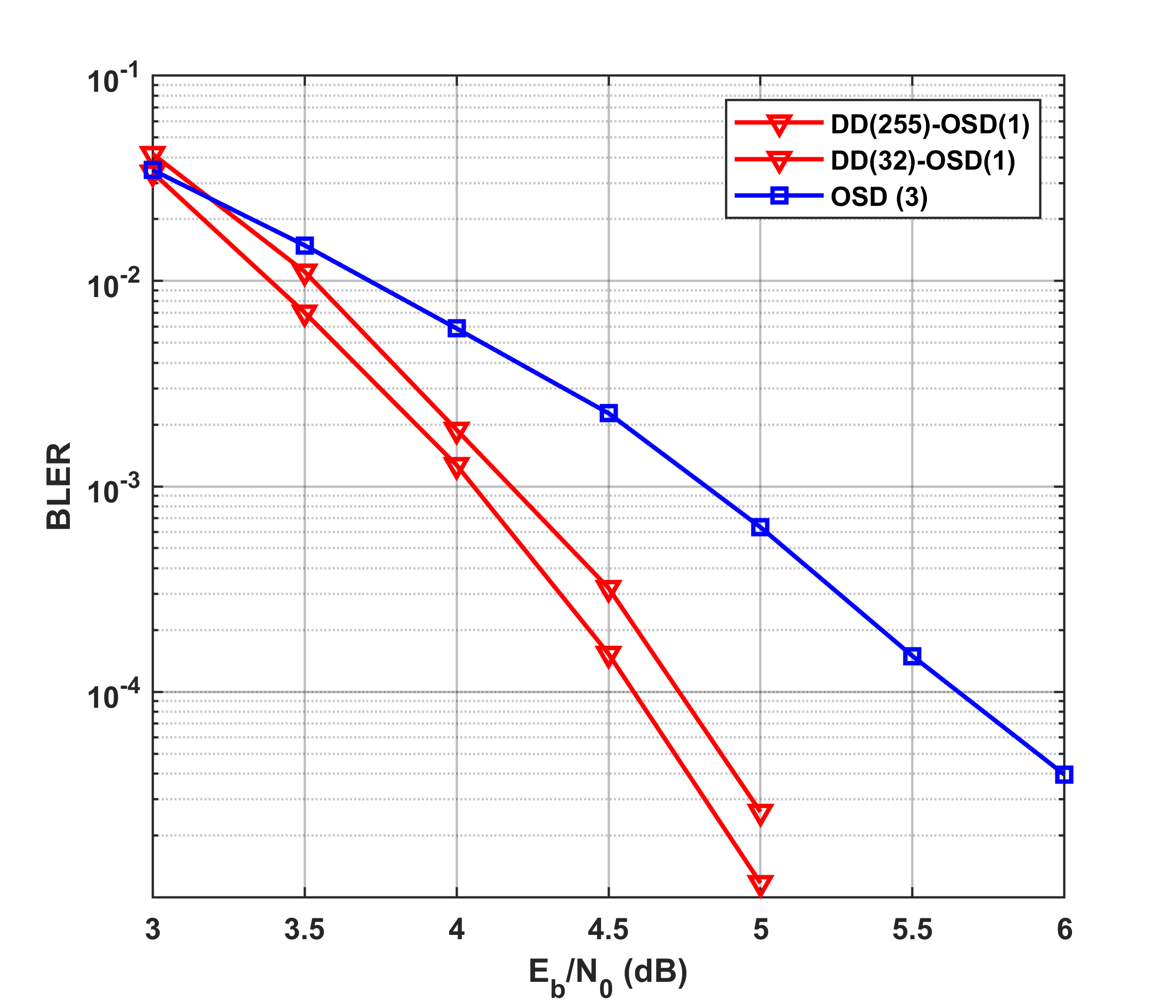}
\end{minipage}
}
\subfigure[$(256, 79)$ eBCH]{%
\begin{minipage}[t]{0.33\linewidth}
\flushleft 
\label{fig:data256_79}%
\includegraphics[width=1.0\textwidth]{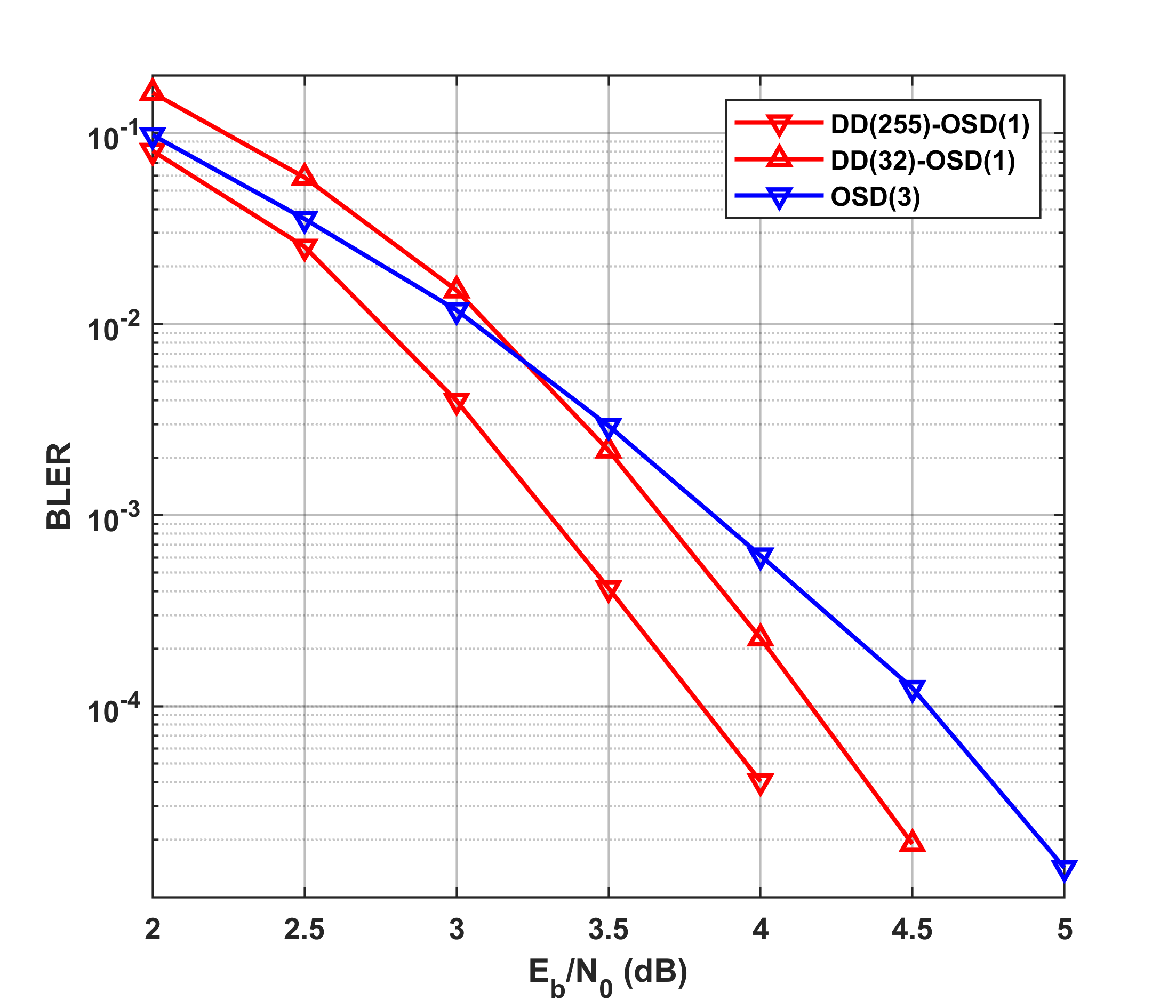}
\end{minipage}
}
\caption{Performance of decoding eBCH codes using DD-OSD$(1)$ and OSD$(3)$.}
\label{fig:DD-OSD eBCH}
\end{figure*}

Consider the $(128, 36)$, $(256, 37)$, and $(256, 79)$ eBCH codes.
We investigate their cyclic DDs and minimal DDs, and list the code paramters in Table \ref{tb: eBCH codes}.
In the top line of Table \ref{tb: eBCH codes}, $n_{\mathcal{C}}$ and $k_{\mathcal{C}}$ denote the code length and the code dimension of the eBCH codes, respectively; $k_{\mathcal{D}}$ and $k_{\mathcal{D}_1}$ denote the code dimension of the cyclic DDs and minimal DDs, respectively; $d_{\mathcal{C},\textbf{BCH}}$ and $d_{\mathcal{D},\textbf{BCH}}$ denote the BCH bounded distance of these eBCH codes and their cyclic DDs, respectively.

In general, DD-OSD$(1)$ outperforms OSD$(3)$ for decoding extended cyclic codes with moderate codelength at high SNR regions.
For derivative decoding the $(n_{\mathcal{C}}, k_{\mathcal{C}})$ eBCH code in Table \ref{tb: eBCH codes}, we take $B$ as a collection of all the nonzero elements in the corresponding splitting field and denote the procedure by DD($|B|$)-OSD($1$). From Remark \ref{remark4}, we can decode its minimal DD as a $(n_{\mathcal{C}}/2, k_{\mathcal{D}_1}/2)$ code with minimal Hamming distance $d_{\mathcal{D},\textbf{BCH}}/2$.
In addition, we perform derivative decoding in $|B| = 32$ directions at random, and denote the procedure by DD$(32)$-OSD$(1)$. 
The maximum iteration number $N_{\emph{DD,max}}$ is $4$ in all the cases.
We compare with the OSD$(3)$ and provide the simulation results over AWGN channels in Fig. \ref{fig:DD-OSD eBCH}.

Following the complexity analysis in \cite{fossorier1995soft}, we investigate the number of floating point operations of OSD of DD-OSD$(1)$. 
It consumes $n_{\mathcal{D}}\text{log}_2(n_{\mathcal{D}})+k_{\mathcal{D}_1}(n_{\mathcal{D}}-k_{\mathcal{D}_1})$ floating point operations to decode the $(n_{\mathcal{D}},k_{\mathcal{D}_1})$ minimal DD using the OSD$(1)$.
From Proposition \ref{prop:complexity}, it consumes $|B|5n + |B|(n_{\mathcal{D}}\text{log}_2\Big(n_{\mathcal{D}})+k_{\mathcal{D}_1}(n_{\mathcal{D}}-k_{\mathcal{D}_1})\Big)$ floating point operations to perform DD($|B|$)-OSD$(1)$ on the $(n,k)$ eBCH code per iteration.

Consider decoding the $(256, 79)$ eBCH code at the $E_b/N_0$ of $4.0$ dB. The average iteration number of DD($255$)-OSD($1$) is $1.02$ and that of DD$(32)$-OSD$(1)$ is $1.03$. As a result, the average cost of DD$(255)$-OSD$(1)$ is $1.02*255*1280+1.02*(255*8 + 255*176*79)= 3,951,439$ floating point operations and that of DD$(32)$-OSD$(1)$ is $1.03*32*1280+1.03*(32*8 + 32*176*79)= 500,728$ floating point operations. For comparision, the cost of OSD$(3)$ is $(\tbinom{79}{3}+\tbinom{79}{2}+\tbinom{79}{1})*176+256*8 = 14,476,112$ floating point operations.

\section{Conclusion}\label{sec:conclusion}
This paper introduces cyclic DDs and minimal DDs for extended cyclic codes and investigates their properties. These properties allow us to decode extended cyclic codes with soft-decision. Besides, it works for cyclic codes of
length of $2^{m} - 1$ as well according to Remark \ref{remark_1}. Simulation results verify that they perform very well for some eBCH codes over AWGN channels. 

{\appendices

\section{Proof of Proposition 4}

The dimension of $\mathcal{A}(\mathcal{C})$ satisfies
$$k_{\mathcal{D}} = |S_{\mathcal{A}}| \leq \sum_{i=0}^{\texttt{deg}(S_{\mathcal{A}})}\binom{m}{i}.$$
From Proposition \ref{prop:ascendant DS}, for any $s \in S_{\mathcal{A}}$, $P(s) \subseteq S_{\mathcal{C}}$.
From (\ref{eq:Ds}), $\texttt{deg}\Big(P(s)\Big) = \texttt{wt}(\overline{s})-1$.
Then
$$\texttt{wt}(\overline{s}) = \texttt{deg}\Big(P(s)\Big) + 1 \leq \texttt{deg}(S_{\mathcal{C}}) + 1.$$
This leads $\texttt{deg}(S_{\mathcal{A}})\leq\texttt{deg}(S_{\mathcal{C}})+1$.
As a result,
$$k_{\mathcal{A}} \leq \sum_{i=0}^{\texttt{deg}(S_{\mathcal{A}})}\binom{m}{i} \leq \sum_{i=0}^{\texttt{deg}(S_{\mathcal{C}})+1}\binom{m}{i}.$$

\section{Proof of Proposition 5}

Let $\mathcal{D}\Big(\mathcal{A}(\mathcal{C})\Big)$ denote the cyclic DD of $\mathcal{A}(\mathcal{C})$ with minimum Hamming distance $d_{\mathcal{D}(\mathcal{A})}$.
From Proposition \ref{prop:distance of dd}, we have $d_{\mathcal{D}(\mathcal{A})} \leq 2d_{\mathcal{A}}$.
From Definition \ref{def:derivative ascendant}, we have $\mathcal{D}\Big(\mathcal{A}(\mathcal{C})\Big) \subseteq \mathcal{C}$
which indicates $d_{\mathcal{D}(\mathcal{A})}\geq d$.
As a result, $d_{\mathcal{A}} \geq d/2$.

\section{Proof of Equivalence for RM Codes}
To begin, we recap the equivalence between representing RM codes by $m$-variate polynomials and representing RM codes by MS polynomials.
The exponent set of RM$(r,m)$ is $S = \{s:0\leq\texttt{wt}(\overline{s})\leq r\}$. 
For any $A(z) \in $RM$(r,m)$, we can write it as
\begin{equation*}
A(z) = \sum_{s\in S}A_sz^s.
\end{equation*}
Note that we can write $z$ as $z=\sum_{i=0}^{m-1}z_i\alpha^i$ where $z_i \in \mathbb{F}_2$ and we can write $s$ as $s=\sum_{j=0}^{m-1}s_j2^j$ where $s_j\in\{0, 1\}$.
Then 
\begin{equation*}
\begin{aligned}
A(z) & = \sum_{s\in S}A_s(\sum_{i=0}^{m-1}z_i\alpha^i)^{\sum_{j=0}^{m-1}s_j2^j} \\
& =  \sum_{s\in S}A_s\prod_{j=0}^{m-1}(\sum_{i=0}^{m-1}z_i\alpha^{i2^j})^{s_j}. \\
\end{aligned}
\end{equation*}
Please note that for any $s\in S$, $0\leq\texttt{wt}(\overline{s})\leq r$.
Thus 
\begin{equation*}
\begin{aligned}
A(z) = \sum_{V\subseteqq [m], |V|\leq r}u_V \prod_{i\in V}z^i,
\end{aligned}
\end{equation*}
where $u_V$ is a summation of $A_s\alpha^{i2^j}$ over a collection of $s,i,j$.
For $|V|=0$, $A(0)=u_{\emptyset}$, so $u_{\emptyset} \in \mathbb{F}_2$.
For $|V|=1$, $A(\alpha^i)=u_{\emptyset}+u_{\{i\}}$. 
Therefore, $u_{\{i\}} \in \mathbb{F}_{2}$ for all $i\in [m]$.
Note that for any $V\subseteqq [m]$, $A(\sum_{i\in V}\alpha^i) = \sum_{V'\subseteqq V}u_{V'} = \sum_{V'\subsetneqq V}u_{V'} + u_V$.
One can easily prove that $u_V \in \mathbb{F}_2$ for all $V\in [m]$ and $|V|\leq r$ by induction.
As a result, we can treat $A(z)$ as a $m$-variate Boolean polynomial $A(z_0, z_1, ..., z_{m-1})$ with degree no larger than $r$. 
Moreover, $[A(\alpha^i), i\in I]$ is equal to $[A(z_0, z_1, ..., z_{m-1}), [z_0, z_1, .., z_{m-1}] \in \mathbb{F}_{2}^m]$.
Note that the dimension of RM$(r,m)$ is $\sum_{i=0}^{r}\tbinom{m}{i}$. 
We conclude that RM$(r,m)$ consists of the evaluation vectors of all the $m$-variate Boolean polynomials with degree no larger than $r$ over $\mathbb{F}_{2}^m$.

First consider the derivative of $A(z)$ in the direction $\alpha^0$,
\begin{equation}
\begin{aligned}
\Delta_{1}A(z) & = A(z+1)-A(z) \\
& = A\Big(\sum_{i=1}^{m-1}z_i\alpha^i+(z_0+1)\Big)-A(\sum_{i=0}^{m-1}z_i\alpha^i) \\
& = \sum_{V\subseteqq [m], |V|\leq r, \atop 0\in V}u_{V}(z_0+1)\prod_{i\in V/\{0\}}z_i + \\
& \sum_{V\subseteqq [m], |V|\leq r, \atop 0\notin V}u_{V}\prod_{i\in V}z_i - \sum_{V\subseteqq [m], |V|\leq r}u_{V}\prod_{i\in V}z_i \\
& = \sum_{V\subseteqq [m], |V|\leq r, \atop 0\in V}u_{V}\prod_{i\in V/\{0\}}z_i.
\end{aligned}
\end{equation}
It is a $(m-1)$-variate Boolean polynomial with degree no larger than $r-1$ and we denote it by $A'(z_1, z_2, ..., z_{m-1})$.
Denote $T_1=\{\sum_{j=1}^{m-1}z_j\alpha^j: z_j \in \mathbb{F}_2 \text{ for } j = 1,2,...,m-1\}$ such that $T_1 \cup (1 + T_1) = \mathbb{F}_{2^m}$. The vector 
\begin{equation}
\begin{aligned}
& [\Delta_{1}A(\alpha^i), \alpha^i \in T] = \\ 
& [A'(z_1,z_2,...,z_{m-1}), [z_1,z_2,...,z_{m-1}] \in (\mathbb{F}_2)^{m-1}]
\end{aligned}
\end{equation}
is a codeword in RM$(r-1,m-1)$.

Now consider the derivative of $A(z)$ in the direction $\beta = \alpha^b$.
Take $T = \{\beta \alpha^i: \alpha^i \in T_1\}$.
We have $\beta+T = \{\beta \alpha^i: \alpha^i \in 1+T_1\}$ and $T\cup (\beta+T) = \mathbb{F}_{2^m}$.
According to Corrollary \ref{coro:coro_1}, the $b$-cyclic shift of $\Delta_{\beta}A(z)$ is equal to $\Delta_{1}A(\beta z)$.
It indicates the evaluation of $\Delta_{\beta}A(z)$ at $\beta\alpha^i$ is equal to the evaluation of $\Delta_{1}A(\beta z)$ at $\alpha^i$.
As a result,  
\begin{equation*}
\begin{aligned}
&[\Delta_{\beta}A(\alpha^i), \alpha^i \in T] = [\Delta_{1}A(\beta\alpha^i), \alpha^i \in T_1\}]
\end{aligned}
\end{equation*}
is also a codeword in RM$(r-1,m-1)$.

\bibliographystyle{IEEEtran}
\bibliography{CyclicCodes.bib}

\end{document}